\documentclass[a4paper,11pt]{amsart}
\usepackage{amssymb}
\usepackage{amscd}
\usepackage{amsmath,amsthm}
\usepackage[colorlinks=true]{hyperref}
\usepackage{enumerate}
\usepackage{booktabs,multirow}
\usepackage{tikz}
\usetikzlibrary{patterns}

\allowdisplaybreaks[1]
\title[Laurent polynomial solutions of the boundary qKZ equation]
{Laurent polynomial solutions of the boundary quantum Knizhnik--Zamolodchikov
equation}
\author[K.~Shigechi]{Keiichi~Shigechi}
\email{k1.shigechi AT gmail.com}
\date{\today}

\newcommand\linkpattern[4]{
\foreach \x/\y in {#1}
\draw[thick](0.8*\x,0)..controls (0.8*\x,-0.6*\y+0.6*\x) 
		and (0.8*\y,-0.6*\y+0.6*\x)..(0.8*\y,0);
\foreach \x/\y in {#2}
\draw[dashed,thick](0.8*\x,0)..controls (0.8*\x,-0.6*\y+0.6*\x) 
		and (0.8*\y,-0.6*\y+0.6*\x)..(0.8*\y,0);
\foreach \x/\y in {#3}
	\draw(0.8*\x,0)--(0.8*\x,-0.5)node[anchor=north]{\rm{\y}};
\foreach \x/\y in {#4}
	\draw(0.8*\x,0)--(0.8*\x,-0.5)
	node[circle,inner sep=1pt,draw,anchor=north]{\rm\y};
}

\newcommand\downa[1]{
\draw(#1,0)--(#1,-0.6)(#1-0.12,-0.6+0.12)--(#1,-0.6)--(#1+0.12,-0.6+0.12);
}
\newcommand\tikzpic[2]{
\raisebox{#1\totalheight}{
\begin{tikzpicture}
#2
\end{tikzpicture}
}}

\newtheorem{theorem}{Theorem}[section]
\newtheorem{example}[theorem]{Example}
\newtheorem{lemma}[theorem]{Lemma}
\newtheorem{defn}[theorem]{Definition}
\newtheorem{prop}[theorem]{Proposition}
\newtheorem{cor}[theorem]{Corollary}

\begin{document}
\begin{abstract}
We construct Laurent polynomial solutions of the boundary quantum 
Knizhnik--Zamolodchikov equation for $U_{q}(\widehat{\mathfrak{sl}}_{2})$ on the parabolic Kazhdan--Lusztig 
bases.
They are characterized by non-symmetric Koornwinder polynomials 
with the specialized parameters. 
As a special case, we obtain the solution of the minimal degree.
\end{abstract}
\maketitle

\section{Introduction}
The quantum Knizhnik--Zamolodchikov (qKZ) equation is 
a system of first order $q$-difference equations \cite{Smi86}, 
which is satisfied by the matrix coefficients of the 
products of the vertex operators in the representation 
theory of quantum affine algebras \cite{FreRes92}.
It was shown that the correlation functions for $XXZ$ spin 
chain with quasi-periodic boundary conditions satisfy
the qKZ equation \cite{JimMiw95,JimMiw96}. 
The qKZ equations reappeared in the Razumov--Stroganov 
correspondence \cite{DiFZJ05-1,RazStr04,RazStrPZJ07}: 
the solution is the ground state of the $XXZ$ spin chain 
at $q$ root of unity .
The qKZ equations for arbitrary root system were obtained 
by Cherednik \cite{Che91,Che92}.
We call them boundary (or reflection) qKZ equations.
In the case of type A, one can construct a solution
by using multidimensional integrals of hypergeometric type
\cite{Mim96,MiwTakTar99,VarTar95}.  
Recently, polynomial solutions were constructed 
by using the representation theory of the affine Hecke algebra.
They are level one for $U_{q}(\widehat{\mathfrak{sl}}_n)$ \cite{DiFZJ05-1,DifZJ05-2}, 
level $-1/2$ for $U_{q}(\widehat{\mathfrak{sl}}_{2})$ \cite{KasPas07} and level $\frac{k+1}{r-1}-n$ 
for $U_q(\widehat{\mathfrak{sl}}_{n})$ \cite{KasTak07}.

The integrability of a physical system with boundaries is ensured by 
the Yang--Baxter equation \cite{Bax82,Yan67} and the reflection 
(or boundary Yang--Baxter) equation \cite{Che84,Skl88}.  
The boundary qKZ equations are classified into three classes depending
on the choice of $K$-matrices (or boundary $R$-matrices).
The first class is the cases where two $K$-matrices are diagonal. 
The boundary qKZ of this class is discussed in 
\cite{deGPya10,DiF05,JimKedKojKonMiw95,JimKedKonMiw95}. 
The second class is the case where one $K$-matrix is diagonal and the other
is not diagonal. 
This class is discussed in \cite{deGPya10,PZJ07}. 
The third class is the case where two $K$-matrices are not diagonal.
In the case of $U_{q}(\widehat{\mathfrak{sl}}_2)$, the boundary qKZ equation of 
the third class has eight parameters: 
three Hecke parameters $q,q_0,q_{N}$, the shift parameter $s$, 
the Baxterization parameter $\zeta_{0}, \zeta_{N}$ and two more 
parameters $\kappa_{0}, \kappa_{N}$. 
The Laurent polynomial solutions in the case of 
$q_{0}=q_{N}=\sqrt{-1}q^{1/2}$ is discussed in \cite{Can09,deGPonShi09}.
In this paper, we consider the boundary qKZ equations of the second 
and the third classes for $U_{q}(\widehat{\mathfrak{sl}}_{2})$.
We call the boundary qKZ equation of the second (resp. third) class
the one-boundary (resp. two-boundary) case.

We construct Laurent polynomial solutions of the boundary qKZ equations
on the parabolic Kazhdan--Lusztig bases for the Hermitian symmetric 
pair $(B_{N},A_{N-1})$ \cite{Boe88,Bre09,Shi14-1} (see also \cite{KL79,Deo87}). 
The boundary qKZ equation can be regarded as compatibility conditions 
of the two representations: a representation by Kazhdan--Lusztig bases 
and a polynomial representation of the affine Hecke algebra. 
The affine Hecke algebra consists of $T_{i}$, $0\le i\le N$, and 
$Y_{i}$, $1\le i\le N$ (see Section \ref{sec-algebras}). 
The Cherednik--Noumi \cite{Nou95} $Y$-operators commute with each other and generalize
Dunkl operators.
In the polynomial representation, the non-symmetric Macdonald--Koornwinder 
polynomials \cite{Koo92} are the simultaneous eigenfunctions of the 
operators $Y_{i}$ \cite{Sah99}. 
The (double) affine Hecke algebra of type $(C^{\vee}_{N}, C_{N})$ with 
$N\ge2$ has six parameters: $q, q_{0}, q_{N}, s, \zeta_{0}$ and $\zeta_{N}$. 
The appearance of the non-symmetric Koornwinder polynomial for generic 
parameters in  the theory of the $XXZ$ spin chain is discussed in \cite{StoVla13}. 
The polynomial representations for specialized parameters are exposed in 
\cite{Kas08}. 
In this paper, we consider the specialization $q^{2(k+1)}s^{2(r'-1)}=1$ 
where 
\begin{eqnarray}
\label{spec-0}
(k,r')=
\begin{cases}
(1,r+1), &\text{for two-boundary case}, \\
(2,2r+1), & \text{for one-boundary case},
\end{cases}
\end{eqnarray}
where $r\in\mathbb{N}_{+}$. 
There, the polynomial representation of the double affine Hecke algebra 
can be non-$Y$-semisimple or reducible. 
In the case of $q^{2(k+1)}s^{2(r-1)}=1$, a polynomial representation considered 
in this paper is characterized by the so-called wheel conditions \cite{FeiJimMiwMuk03,Kas05}. 
On the other hand, there exists a Kazhdan--Lusztig basis such that the basis 
is a simultaneous eigenfunction of the $Y$-operators.
Thus we can identify the two representations by finding the non-symmetric 
Koornwinder polynomial associated with the Kazhdan--Lusztig basis.
If the paramters satisfy 
\begin{eqnarray*}
&&q_{N}^{2}=-q, \qquad\qquad\text{for one-boundary case}, \\
&&\omega_{m}^{\pm mJ/r}q^{\mp(N-1+4J/r)}(q_{0}q_{N})^{\pm1}=\kappa_{0}\kappa_{N}, 
\quad \text{for two-boundary case}
\end{eqnarray*}
where $\omega_{m}$ is a primitive $m$-th root of unity, $m=GCD(k+1,r-1)$ 
and $J,r\in\mathbb{N}_{+}$,  
we have a Laurent polynomial solution characterized by a non-symmetric Koornwinder 
polynomial with the specialized parameters.
As a special case, we obtain the solution of the minimal degree.

The plan of the paper is as follows. 
We briefly review the affine Hecke algebra and the two-boundary 
Temperley--Lieb algebra in Section \ref{sec-algebras}.
In Section \ref{sec-Reps}, we summarize the representations of the 
one- and two-boundary Temperley--Lieb algebra on the standard bases and 
on the parabolic Kazhdan--Lusztig bases. 
We collect the definitions and properties of non-symmetric Koornwinder 
polynomials under the specialization (\ref{spec-0}).
In section \ref{sec-qKZ}, we recall the boundary qKZ equations for 
one- and two-boundary cases.
In Section \ref{sec-redqKZ}, we consider the reduction of the boundary
qKZ equation. 
Since the boundary qKZ equations contain many equivalent equations, 
we extract non-trivial equations from them.
We also show the conditions which a solution of the boundary qKZ equation
satisfies.
In Section \ref{sec-Laurent}, we construct Laurent polynomial solutions of 
the boundary qKZ equation by combining the results of Section \ref{sec-Reps} 
and Section \ref{sec-redqKZ}.

\section{Algebras}
\label{sec-algebras}
\subsection{Affine Hecke algebra}

The affine root system of type $C_{N}$ can be 
realized in $\mathbb{R}^{N}\oplus\mathbb{R}\delta$
where $\delta$ is a radical element. 
The affine simple roots are given by 
\begin{eqnarray*}
\alpha_i=\epsilon_i-\epsilon_{i+1},\quad (1\le i\le N-1), \quad
\alpha_N=2\epsilon_N, \quad \alpha_0=\delta-2\epsilon_1,
\end{eqnarray*}
where $\epsilon_i$ is the standard 
orthonormal bases satisfying 
$\langle\epsilon_i,\epsilon_j\rangle=\delta_{ij}$.
We denote by 
$\alpha_{i}^{\vee}:=2\alpha_i/\langle\alpha_i,\alpha_i\rangle$
simple co-roots.

Let $W_0:=\langle s_1,\ldots,s_N\rangle$ and 
$W:=\langle s_0,\ldots,s_N\rangle$ be the finite 
and affine Weyl group of type $C_N$ respectively.
The faithful action of $W$ on $\mathbb{R}^{N}$ is given 
by 
\begin{eqnarray*}
&&s_0\cdot(v_1,\ldots,v_N)=(-1-v_1,v_2,\ldots,v_N), \\
&&s_i\cdot(v_1,\ldots,v_N)=(v_1,\ldots,v_{i-1},v_{i+1},v_{i},
v_{i+2},\ldots,v_N), \quad 1\le i\le N-1, \\
&&s_N\cdot(v_1,\ldots,v_N)=(v_1,\ldots,v_{N-1},-v_N), 
\end{eqnarray*}
The $s$-dependent action of the affine Weyl group on 
$\mathbf{z}:=(z_1,\ldots,z_N)\in(\mathbb{C}^{*})^{N}$
is given by 
\begin{eqnarray*}
s_0\mathbf{z}&:=&(s^{2}z_1^{-1},z_2,\ldots,z_N), \\
s_i\mathbf{z}&:=&(z_1,\ldots,z_{i-1},z_{i+1},z_{i},
z_{i+2},\ldots,z_{N}), \quad 1\le i\le N-1,  \\
s_N\mathbf{z}&:=&(z_1,\ldots,z_{N-1},z_{N}^{-1}).
\end{eqnarray*}

The affine Hecke algebra $\mathcal{H}_N$ is 
the unital associative $\mathbb{C}(q,q_0,q_{N})$-algebra generated 
by $T_i$ ($0\le i\le N$) and $Y_j$ ($1\le j\le N$) satisfying
the relations:
\begin{eqnarray*}
&&(T_i+q)(T_i-q^{-1})=0,\quad 1\le i\le N-1, \\ 
&&(T_N+q_N)(T_N-q_N^{-1})=0, \\
&&(T_0+q_0)(T_0-q_0^{-1})=0, \\
&&T_0T_1T_0T_1=T_1T_0T_1T_0, \\
&&T_{N}T_{N-1}T_{N}T_{N-1}=T_{N-1}T_{N}T_{N-1}T_N, \\
&&T_iT_{i+1}T_{i}=T_{i+1}T_{i}T_{i+1}, \quad 1\le i\le N-2, \\
&&T_iT_j=T_jT_i, \quad |i-j|>1, \\
&&Y_iY_j=Y_jY_i, \quad \forall i,j,\\
&&T_{i}Y_{j}=Y_jT_i, \quad \langle \alpha_i,\epsilon_j\rangle=0,\\
&&T_iY_{i+1}T_{i}=Y_{i}, \quad 1\le i\le N-1,  \\
&&T_N^{-1}Y_N=Y_{N}^{-1}T_N-(q_0-q_0^{-1}).
\end{eqnarray*}
By the Bernstein--Zelevinsky presentation of the affine Hecke algebra 
(see, {\it e.g.}, \cite{Lus89}), 
we have the correspondence 
\begin{eqnarray}
\label{BZ-TY}
Y_i\mapsto T_i\ldots T_{N-1}T_N\ldots T_0T_1^{-1}\ldots T_{i-1}^{-1}.
\end{eqnarray}
We denote by $\mathcal{H}_{N}^{0}$ the Hecke algebra of type B generated 
by $T_{i}, 1\le i\le N$.

\subsection{Two-boundary Temperley--Lieb algebra}
The {\it two-boundary Temperley--Lieb} algebra
\cite{deGNic09,deGPya04} 
is a unital associative algebra over 
$\mathbb{C}(q,q_{0},q_{N},\kappa_{0},\kappa_{N})$ 
generated by $e_i$, $0\le i\le N$, with the 
relations
\begin{eqnarray*}
&&e_i^2=-(q_i+q_i^{-1})e_i,\quad 0\le i\le N, \\
&&e_ie_{i\pm1}e_i=a_{i,i\pm1}e_i, \quad 1\le i\le N-1, \\
&&e_ie_j=e_je_i, \quad |i-j|>1,
\end{eqnarray*}
where $q_1=q_2=\ldots=q_{N-1}:=q$ and 
\begin{eqnarray}
a_{i,j}:= 
\begin{cases}
\kappa_{N}(qq_{N}^{-1}+q^{-1}q_N), & (i,j)=(N-1,N), \\
\kappa_{0}(qq_{0}^{-1}+q^{-1}q_0), &  (i,j)=(1,0), \\
1, & otherwise.
\end{cases}
\end{eqnarray}
The two-boundary Temperley--Lieb algebra is infinite dimensional.
We consider the following two conditions to make the algebra finite 
dimensional: 
\begin{eqnarray*}
I_{N}J_{N}I_{N}=\alpha I_N, \qquad
J_{N}I_{N}J_{N}=\alpha J_N
\end{eqnarray*}
where 
\begin{eqnarray*}
&&I_{2n}:=\prod_{i=0}^{n-1}e_{2i+1}, \quad I_{2n+1}:=e_{0}\prod_{i=1}^{n}e_{2i}, \\
&&J_{2n}:=e_{0}\prod_{i=1}^{n-1}e_{2i}\cdot e_{N}, \quad 
J_{2n+1}:=\prod_{i=0}^{n-1}e_{2i+1}\cdot e_{N}, \\
&&\alpha=\begin{cases}
(\kappa_{N}^{-1}+\kappa_{0}q_{0}q_{N}^{-1})
(\kappa_{0}^{-1}+\kappa_{N}q_{0}^{-1}q_{N}), & \text{for $N$  odd}, \\
(\kappa_{0}q_{N}^{-1}-\kappa_{N}^{-1}q^{-1}q_{0})
(\kappa_{N}q_{N}-\kappa_{0}^{-1}qq_{0}^{-1}), & \text{for $N$ even}.
\end{cases}
\end{eqnarray*}

The subalgebra generated by $e_1,\ldots, e_{N}$ is the 
{\it one-boundary Temperley--Lieb} algebra. 

\section{Representations}
\label{sec-Reps}
\subsection{Two-boundary Temperley--Lieb algebra}
Denote $V\cong\mathbb{C}^{2}$ be a $\mathbb{C}$-vector space with 
an ordered basis $(v_+,v_-)$. 
We consider the lexicographic order of basis in $V^{\otimes m}$.
For example, the ordered bases in $V^{\otimes2}$ are 
$(v_+\otimes v_+,v_+\otimes v_-,v_-\otimes v_+,v_-\otimes v_-)$.  
Let $\epsilon\in\{+,-\}^{N}$ be a binary string of length $N$. 
We abbreviate by $v_{\epsilon}$ a basis 
$v_{\epsilon_1}\otimes \ldots\otimes v_{\epsilon_N}$

The two-boundary Temperley--Lieb algebra has a natural representation 
in $\mathrm{End}_{\mathbb{C}}(V^{\otimes N})$
(see, {\it e.g.}, \cite{deGNic09} and references therein). 
This representation has two parameters $\kappa_0$ and $\kappa_N$.
The matrix representation of the generators are 
\begin{eqnarray*}
e_i&=&\underbrace{\mathbf{1}\otimes\ldots\otimes\mathbf{1}}_{i-1}
\otimes 
\begin{pmatrix}
0 & 0 & 0 & 0 \\
0 & -q^{-1} & 1 & 0 \\
0 & 1 & -q & 0 \\
0 & 0 & 0 & 0 
\end{pmatrix}
\otimes\underbrace{\mathbf{1}\otimes\cdots\otimes\mathbf{1}}_{N-i-1}, 
\qquad 1\le i\le N-1,\\
e_N&=&\underbrace{\mathbf{1}\otimes\ldots\otimes\mathbf{1}}_{N-1}
\otimes 
\begin{pmatrix}
-q_N^{-1} & \kappa_N \\
\kappa^{-1}_N & -q_N 
\end{pmatrix}, \\
e_0&=&
\begin{pmatrix}
-q_0 & \kappa_0^{-1} \\
\kappa_0 & -q_0^{-1} 
\end{pmatrix} 
\otimes 
\underbrace{\mathbf{1}\otimes\ldots\otimes\mathbf{1}}_{N-1}.
\end{eqnarray*}
 
We define the $R$-matrix acting on $V\otimes V$ by
\begin{eqnarray}
\check{R}_i(z):=\frac{qz-q^{-1}}{q-q^{-1}z}\mathbf{1}+\frac{z-1}{q-q^{-1}z}e_i.
\end{eqnarray}
The $R$-matrix satisfies the unitarity condition 
$\check{R}(z)\check{R}(1/z)=\mathbf{1}$
and the Yang--Baxter equation \cite{Bax82,Yan67}:
\begin{eqnarray*}
\check{R}_{i}(z)\check{R}_{i+1}(zw)\check{R}_{i}(w)
=\check{R}_{i+1}(w)\check{R}_{i}(zw)\check{R}_{i+1}(z).
\end{eqnarray*}
Similarly, we define the $K$-matrices (or boundary $R$-matrices) as 
\begin{eqnarray*}
K_N(z)&:=&
\frac{(z+q_{N}\zeta_{N})(z-q_{N}\zeta_{N}^{-1})}
{(1+q_{N}\zeta_{N}z)(1-q_{N}\zeta_{N}^{-1}z)}\mathbf{1}
+
\frac{q_{N}(z^2-1)}
{(1+q_{N}\zeta_{N}z)(1-q_{N}\zeta_{N}^{-1}z)}e_{N}, \\
K_{0}(z)&:=&
\frac{(z^{-1}+q_{0}\zeta_{0})(z^{-1}-q_{0}\zeta_{0}^{-1})}
{(1+q_{0}\zeta_{0}z^{-1})(1-q_{0}\zeta_{0}^{-1}z^{-1})}\mathbf{1}
+
\frac{q_{0}(z^{-2}-1)}
{(1+q_{0}\zeta_{0}z^{-1})(1-q_{0}\zeta_{0}^{-1}z^{-1})}e_{0}, \\
\end{eqnarray*}
where $\zeta_0$ and $\zeta_N$ are free parameters which appear 
by the Baxterization.
These $K$-matrices satisfy the unitarity equation 
$K_0(z)K_{0}(1/z)=1=K_N(z)K_{N}(1/z)$ 
and the boundary Yang--Baxter (or reflection)
equations \cite{Che84,Skl88}:
\begin{eqnarray*}
K_{N}(w)\check{R}_{N-1}(1/(zw))K_{N}(z)\check{R}_{N-1}(w/z)
&=&
\check{R}_{N-1}(w/z)K_{N}(z)\check{R}_{N-1}(1/(wz))K_{N}(w), \\
K_{0}(z)\check{R}_{1}(zw)K_{0}(w)\check{R}_{1}(w/z)
&=&
\check{R}_{1}(w/z)K_{0}(w)\check{R}_{1}(wz)K_{0}(z). \\
\end{eqnarray*}

\subsection{Kazhdan--Lusztig bases}
We consider the representation of the affine Hecke algebra 
in $V^{\otimes N}$.
The two-boundary Temperley--Lieb algebra can be regarded 
as the affine Hecke algebra with quotient relations 
through $T_i\mapsto e_i+q_i^{-1}$. 
We will define three types of (parabolic) Kazhdan--Lusztig 
bases of $\mathcal{H}_{N}^{0}$ following \cite{Deo87,KL79,Lus03,Shi14-1}. 
We call these bases type BI, BII and BIII respectively. 
The Hecke parameters satisfy $q_{N}=q^{M}$, $M\in\mathbb{Z}_{\ge1}$,  
for type BI and $q_N$ and $q$ are algebraically independent 
for type BII and BIII. 
We consider the abelian groups 
$A^{I}:=\{q^{i}\kappa_N^{j}|i\in\mathbb{Z}, j\in\mathbb{Z}_{\ge0}\}$ 
for type BI and 
$A^{II}:=\{q^{i}q_N^{j}\kappa_N^{k}|i,j\in\mathbb{Z},k\in\mathbb{Z}_{\ge0}\}
=:A^{III}$ for type BII and BIII.
The lexicographic order of $A^{X}$ (X=I, II, III) is defined 
by $A^{X}=A_{+}^{X}\cup\{\kappa_{N}^{i}| i\in\mathbb{Z}_{\ge0}\}\cup A_{-}^{X}$ 
where 
\begin{eqnarray*}
A^{I}_{+}&:=&\{q^{i}\kappa_N^{j}|i>0, j\ge0\}, \\
A^{II}_{+}
&:=&\{q^{i}q_{N}^j\kappa_N^{k}|i>0, j\in\mathbb{Z}, k\ge0\}\cup\{q_N^{i}\kappa_N^{j}|i>0, j\ge0\}, \\
A^{III}_{+}
&=&\{q^{i}q_{N}^{j}\kappa_N^{k}|i\in\mathbb{Z}, j>0, k\ge0\}\cup\{q^{i}\kappa_N^{j}|i>0, j\ge0\}. 
\end{eqnarray*}
We define the involutive ring automorphism of $\mathcal{H}^{0}_N$, 
known as the bar involution, via $\overline{T_{i}}=T_{i}^{-1}$, 
$\overline{q_i}=q_{i}^{-1}$ for $1\le i\le N$. 
The action of $\mathcal{H}^{0}$ on $V^{\otimes N}$ can be realized by 
$T_{i}\mapsto e_{i}+q_{i}^{-1}$ for $0\le i\le N$. 
On the module $V^{\otimes N}$, we define 
$\overline{\kappa_N}=\kappa_N$ and 
$\overline{v_{b_+}}=v_{b_+}$ where $b_+=(+\ldots+)$.

The (parabolic) Kazhdan--Lusztig bases $C_{\epsilon}^{X}$, $\epsilon\in\{\pm\}^{N}$, 
$X=I,II$ or $III$, satisfy 
\begin{enumerate}
\item $\overline{C^{X}_{\epsilon}}=C^{X}_{\epsilon}$, 
\item 
$C^{X}_{\epsilon}=v_{\epsilon}+\sum_{\epsilon'<\epsilon}
\mathbb{Z}(A_{-}^{X})v_{\epsilon'}$ for $X=I,II$ or $III$.
\end{enumerate}
If we set $\kappa_N=1$, we have standard parabolic 
Kazhdan--Lusztig bases studied in \cite{Shi14-1}. 

We briefly review a graphical presentation of Kazhdan--Lusztig 
bases following \cite{Shi14-1}. 
Let $b$ be a binary string $b\in\{\pm\}^{N}$. 
We place an up arrow (resp. a down arrow) from left to right 
according to $b_i=+$ (resp. $b_i=-$). 
First we make a pair between a down arrow and an up arrow
which are next to each other and in this order. 
We connect this pair of arrows via an arc.
Repeat this procedure until all the up arrows are to the 
left to down arrows.
We have three cases according to the type of the lexicographic
order of $A^{X}$.
\paragraph{\bf Type BI}
We put an integer $p$, $1\le p\le M$, on the $(M+1-p)$-th down arrow
from right.
For remaining down arrows, we make a pair of adjacent down arrows from 
right to left. Then we connect the pair via a dashed arc.

\paragraph{\bf Type BII}
We put a vertical line with a mark e (resp. o) on the $2i$-th (resp.
$(2i-1)$-th) down arrow from right. 

\paragraph{\bf Type BIII}
We put a vertical line with a circled integer $i$ on the 
$i$-th down arrow from right to left.

A partial diagram corresponds to a vector in $V^{\otimes N}$ as follows: 
\begin{eqnarray*}
\raisebox{-0.6\totalheight}{
\begin{tikzpicture}
\draw(0,0)..controls(0,-0.6)and(0.8,-0.6)..(0.8,0);
\end{tikzpicture}}\ 
&=&v_{-+}-q^{-1}v_{+-}, \\
\raisebox{-0.6\totalheight}{
\begin{tikzpicture}
\draw[dashed](0,0)..controls(0,-0.6)and(0.8,-0.6)..(0.8,0);
\end{tikzpicture}}\ 
&=&v_{--}-\kappa_N^{2}q^{-1}v_{++}, \\
\raisebox{-0.6\totalheight}{
\begin{tikzpicture}
\draw(0,0)--(0,-0.6)node[anchor=north]{p};
\end{tikzpicture}}\ 
&=&v_{-}-\kappa_{N}q^{-p}v_{+},  \\
\raisebox{-0.6\totalheight}{
\begin{tikzpicture}
\draw(0,0)--(0,-0.6)node[anchor=north]{o};
\end{tikzpicture}}\ 
&=&v_{-}-\kappa_{N}q_{N}^{-1}v_{+},  \\
\raisebox{-0.6\totalheight}{
\begin{tikzpicture}
\draw(0,0)--(0,-0.6)node[anchor=north]{e};
\end{tikzpicture}}\ 
&=&v_{-}+\kappa_{N}q^{-1}q_{N}v_{+},  \\
\raisebox{-0.6\totalheight}{
\begin{tikzpicture}
\draw(0,0)--(0,-0.6)node[circle,inner sep=0pt,draw,anchor=north]{p};
\end{tikzpicture}}\ 
&=&v_{-}-\kappa_{N}q^{p-1}q_{N}^{-1}v_{+},  \\
\end{eqnarray*}
and an unpaired up (resp. down) arrow corresponds to $v_{+}$ (resp. $v_{-}$).
Since a diagram for $b\in\{\pm\}^{N}$ can be regarded as a tensor product 
of the building blocks, we obtain a vector in $V^{\otimes N}$ by 
tensoring the above expressions.

\begin{example}
Let $\epsilon=(--+-)$. 
\begin{eqnarray*}
C_{\epsilon}^{I}&=&
\tikzpic{-0.6}{
\linkpattern{0.8/1.6}{}{2.1/1}{};
\downa{0.3};
}=v_{--+-}-q^{-1}v_{-+--}-\kappa_{N}q^{-1}v_{--++}
+\kappa_{N}q^{-2}v_{-+-+}, \\
C_{\epsilon}^{II}&=&
\tikzpic{-0.6}{
\linkpattern{0.8/1.6}{}{2.1/o,0.3/e}{};
}  \\ 
&=&v_{--+-}-q^{-1}v_{-+--}-\kappa_{N}q_N^{-1}v_{--++}
+\kappa_Nq^{-1}q_Nv_{+-+-}+\kappa_{N}q^{-1}q_N^{-1}v_{-+-+} \\
&&
-\kappa_Nq^{-2}q_Nv_{++--}-\kappa_N^{2}q^{-1}v_{+-++}
+\kappa_N^{2}q^{-2}v_{++-+}, 
\\
C_{\epsilon}^{III}&=&
\tikzpic{-0.6}{
\linkpattern{0.8/1.6}{}{}{2.1/1,0.3/2}
} \\
&=&v_{--+-}-q^{-1}v_{-+--}-\kappa_Nq_N^{-1}v_{--++}
-\kappa_{N}qq_N^{-1}v_{+-+-}+\kappa_{N}q^{-1}q_{N}^{-1}v_{-+-+} \\
&&+\kappa_{N}q_N^{-1}v_{++--}+\kappa_N^{2}qq_N^{-2}v_{+-++}
-\kappa_{N}^{2}q_N^{-2}v_{++-+}.
\end{eqnarray*}
\end{example}

The actions of the two-boundary Temperley--Lieb algebra in the case of 
$\kappa_{0}, \kappa_{N}=1$ are summarized in \cite[Section 3]{Shi14-3}. 
Due to the existence of $\kappa_{0}$ and $\kappa_{N}$, a slight modification is 
necessary. 
For example, the action of $e_i$, $1\le i\le N-1$, on the Kazhdan--Lusztig basis 
of type BII is given by 
\begin{eqnarray*}
e_{i}\left(\tikzpic{-0.5}{
\linkpattern{}{}{0/e,0.6/o}{}
}\right)&=&
k_{N}(qq^{-1}_{N}+q^{-1}q_{N})\tikzpic{-0.5}{
\linkpattern{0/0.8}{}{}{}
}, \\
e_{i}\left(\tikzpic{-0.5}{
\linkpattern{}{}{0/o,0.6/e}{}
}\right)&=&
-k_{N}(q_{N}+q^{-1}_{N})\tikzpic{-0.5}{
\linkpattern{0/0.8}{}{}{}
}, \\
e_{N}\left(\tikzpic{-0.5}{
\linkpattern{0/0.8}{}{}{}
}\right)&=&k_{N}^{-1}\tikzpic{-0.5}{
\linkpattern{}{}{0/e,0.6/o}{}
}.
\end{eqnarray*}

\subsection{One-boundary Temperley--Lieb algebra}
The space $V^{\otimes N}$ is reducible as a representation 
of the one-boundary Temperley--Lieb algebra. 
Let $\mathcal{B}_N$ be a set of binary strings 
$b=(b_1\ldots b_N)\in\{\pm\}^{N}$
satisfying $\sum_{i=1}^{j}b_j\le0$ for all $1\le j\le N$.
Let $\mathcal{L}_{N}$ be a vector space spanned by 
$\{C^{II}_{\epsilon}| \epsilon\in\mathcal{B}_{N}\}$.
Then, the space $\mathcal{L}_{N}$ is irreducible as a representation of 
the one-boundary Temperley--Lieb algebra.
The dimension of $\mathcal{L}_{N}$ is 
$\displaystyle\mathrm{dim}(\mathcal{L}_{N})=\genfrac{(}{)}{0pt}{}{N}{\lfloor N/2\rfloor}$.

In the case of Type BII, a basis $C^{II}_{\epsilon}\in\mathcal{L}_{N}$ 
is studied as a link pattern with a boundary in literatures (see, {\it e.g.},
\cite{deG05,PZJ07}).

\subsection{The non-symmetric Koornwinder polynomials}
Let $\mathbb{K}:=\mathbb{C}(s,q,q_N,q_0,\zeta_N,\zeta_0)$ 
and $P_N:=\mathbb{K}[z_1^{\pm1},\ldots,z_N^{\pm1}]$ be the ring of 
$N$-variable Laurent polynomials.
Following \cite{Nou95}, we define linear operators $\hat{T}_{0},\ldots,\hat{T}_{N}$ 
on $P_N$ by 
\begin{eqnarray*}
\hat{T}_{0}&:=&-q_0-q_0^{-1}
\frac{(1-sq_0\zeta_{0}^{-1}z_1^{-1})(1+sq_0\zeta_0z_{1}^{-1})}{1-s^{2}z_1^{-2}}(s_0-1)  \\
\hat{T}_i&:=&-q-q^{-1}\frac{1-q^2z_{i}z_{i+1}^{-1}}{1-z_{i}z_{i+1}^{-1}}(s_i-1),  
\quad 1\le i\le N-1,  \\
\hat{T}_N&:=&-q_N-q_N^{-1}
\frac{(1+q_N\zeta_Nz_{N})(1-q_N\zeta_N^{-1}z_{N})}{1-z_N^{2}}(s_N-1).
\end{eqnarray*}
The map $T_i\mapsto \hat{T}_i, Y_i\mapsto \hat{Y}_i$ 
gives the polynomial representation of the (double) affine Hecke algebra.

Fix an element $\lambda\in\mathbb{Z}^{N}$.
Let $\lambda^{+}$ be a unique dominant element in $W_0\lambda$ 
and $w_{\lambda}^{+}$ be the shortest element in $W_0$ such 
that $w_{\lambda}^{+}\lambda^{+}=\lambda$.
We define $\rho(\lambda):=w_{\lambda}^{+}\rho$ where 
$\rho=(N-1,N-2,\ldots,0)$ and 
$\sigma(\lambda)=(\mathrm{sign}(\lambda_1),\ldots,\mathrm{sign}(\lambda_N))$.
We define the dominance order $\ge$ and a partial order $\succeq$ 
on $\mathbb{Z}^{N}$ as follows.
We denote by $\lambda\ge\mu$ if 
$\lambda-\mu\in\sum_{i=1}^{N}\mathbb{Z}_{\ge0}\alpha_{i}^{\vee}$ 
and by $\lambda\succeq\mu$ if $\lambda^{+}>\mu^{+}$, or 
$\lambda^{+}=\mu^{+}$ and $\lambda\ge\mu$.

Set $\lambda\in\mathbb{Z}^{N}$ and 
$\mathbf{z}^{\lambda}:=z_1^{\lambda1}\ldots z_{N}^{\lambda_N}$. 
The {\it non-symmetric Macdonald--Koornwinder} polynomial \cite{Sah99} 
$E_{\lambda}(\mathbf{z};s^2,q^2)$ with parameters $s^{2}$
and $q^{2}$ is a Laurent polynomial satisfying 
\begin{eqnarray*}
\hat{Y}_{i}E_{\lambda}&=&y(\lambda)_iE_\lambda, \\
E_{\lambda}&=&\mathbf{z}^{\lambda}
+\sum_{\mu\prec\lambda}c_{\lambda\mu}\mathbf{z}^{\mu}, 
\quad c_{\lambda\mu}\in\mathbb{K},
\end{eqnarray*}
where 
\begin{eqnarray*}
y(\lambda)_i:=s^{2\lambda_i}q^{2\rho(\lambda)_i}(q_0q_N)^{\sigma(\lambda)_i}.
\end{eqnarray*}

In this paper, we are interested in a specialization of parameters \cite{Kas05,Kas08}: 
\begin{eqnarray}
\label{spec}
s^{2(r'-1)}q^{2(k+1)}=1,
\end{eqnarray}
where $N\ge k+1\ge0$ and $r'-1\ge1$.
More precisely, we consider 
\begin{eqnarray*}
s^{2(r'-1)/m}q^{2(k+1)/m}=\omega_{m}, 
\end{eqnarray*}
where $m=GCD(k+1,r-1)$ and $\omega_{m}$ is a 
primitive $m$-th root of unity. 
Recall that we have the one- and two-boundary Temperley--Lieb
algebras as quotient algebras of the affine Hecke algebra of type C.
We consider the following specialization:  
\begin{eqnarray*}
(k,r')&=&(1,r+1) \text{ for two-boundary case}, \\
(k,r')&=&(2,2r+1) \text{ for one-boundary case},
\end{eqnarray*}
where $r\ge\mathbb{N}_{+}$. 


For a positive integer $J\in\mathbb{N}_{+}$, we define 
$\nu^{J,\pm}:=(\nu^{J,\pm}_{1},\ldots,\nu^{J,\pm}_{N})\in\mathbb{Z}^{N}$ 
and $\xi^{0},\xi^{1}\in\mathbb{Z}^{N}$ 
by 
\begin{eqnarray*}
\nu^{J,+}_{i}&:=&J+r(N-i), \text{ for $1\le i\le N$}\\
\nu^{J,-}_{i}&:=&-J-r(i-1),\text{ for $1\le i\le N$} \\
\xi^{0}_{i}&:=&
\begin{cases}
2(\lfloor(N+1)/2\rfloor-i)r, & 1\le i\le\lfloor(N+1)/2\rfloor, \\
(2N-2i+1)r, & \lfloor(N+1)/2\rfloor+1\le i\le N,
\end{cases}\\
\xi^{1}_{i}&:=&
\begin{cases}
(N-2i)r, & 1\le i\le (N-1)/2, \\
(2N-2i)r, & (N+1)/2\le i\le N  
\end{cases}
\text{ for $N$ odd},
\end{eqnarray*}
Note that the dominant element 
$\xi^{+}:=(\xi^{+}_{1},\ldots,\xi^{+}_{N})\in\mathbb{Z}_{\ge0}$ 
associated with $\xi^{0}$ and $\xi^{1}$ is  
$\xi^{+}_{i}=(N-i)r$.

Recall the definition of admissibility in the case of $s^{2(r'-1)}q^{2(k+1)}=1$: 
\begin{defn}[Definition 4.3 in \cite{Kas08}]
\label{Def-adm}
Fix $(k,r')$ and $\lambda\in\mathbb{Z}^{N}$. 
A pair $(i,j)$ is said to be a $(k+1,r'-1)$-{\it neighbourhood} if 
$(i,j)$ satisfies the following three conditions:
\begin{enumerate}
\item $|\rho(\lambda)_i|-|\rho(\lambda)_j|=k$, 
\item $|\lambda_i|-|\lambda_j|\le r'-1$, 
\item if $|\lambda_i|-|\lambda_j|=r'-1$, then the sign 
$\sigma(\lambda)$ satisfies one of the following three conditions:
\begin{enumerate}
\item $(\sigma(\lambda)_i,\sigma(\lambda)_j)=(+,+)$ and $i>j$, 
\item $(\sigma(\lambda)_i,\sigma(\lambda)_j)=(-,-)$ and $i<j$,
\item $(\sigma(\lambda)_i,\sigma(\lambda)_j)=(-,+)$.
\end{enumerate}
\end{enumerate}
We call $\lambda$ admissible if there is no neighbourhood pairs 
in $\lambda$.
\end{defn}
\begin{defn}
Take two elements $\lambda_1,\lambda_2\in\mathbb{Z}^{N}$. 
The element $\lambda_2$ is said to be connected to $\lambda_1$ if and 
only if there exist a sequence of admissible elements 
$\lambda^{(0)}=\lambda_{1}, \lambda^{(1)}, \ldots, \lambda^{(l)}=\lambda_2$
and a sequence of integers $i_1,\ldots, i_l$ such that 
$\lambda^{(j)}=s_{i_j}\lambda^{(j-1)}$ with $s_{i_j}\in W_{0}$.
\end{defn}

Suppose that an element $\mu$ is admissible.
\begin{defn}
We denote by $I^{(k,r')}(\mu)$ a vector space spanned by 
\begin{eqnarray*}
\{E_{\lambda}|\lambda\in W_0\mu^{+},
 \lambda \text{ is admissible and connected to $\mu$}\}
\end{eqnarray*} 
at $s^{2(r'-1)}q^{2(k+1)}=1$.
\end{defn}

\begin{theorem}[Theorem 4.6 in \cite{Kas08}]
\label{thrm-AHA}
The space $I^{(k,r')}(\mu)$ is a representation 
of the affine Hecke algebra.
\end{theorem}
The space spanned by $\{E_{\lambda}|\lambda \text{ is admissible}\}$ 
is an irreducible representation of the double affine Hecke 
algebra \cite[Theorem 4.6]{Kas08}.
In general, the space $I^{(k,r')}(\mu)$ may be reducible as a 
representation of the affine Hecke algebra. 
However, one can construct an irreducible representation from 
$I^{(k,r')}(\mu)$. 

\subsubsection{Two-boundary case}
\begin{prop}
\label{prop-dim}
The dimension of the space $I^{(1,r+1)}(\nu^{J,+})$ is $2^{N}$.
\end{prop}
\begin{proof}
Since $I:=I^{(1,r+1)}(\nu^{J,+})$ is spanned by admissible elements 
such that $\nu\in W_0\nu^{J,+}$, the dimension of $I$
is equal to the number of admissible elements.
We prove Proposition by induction on $N$. 
When $N=2$, admissible elements are $(J+r,J), (J+r,-J), (-J,J+r)$ 
and $(-J,-J-r)$, which implies $\mathrm{dim}(I)=4$.
We assume that Proposition holds true up to $N-1$ and 
$\mathrm{dim}(I)=2^{N-1}$.

Let $\nu\in W_0\nu^{J,+}$. 
Suppose $\nu_1$ is neither $J+(N-1)r$ nor $-J$. 
When $\mathrm{sign}(\nu_1)=+$, there exists $i$, $2\le i\le N$ 
such that $|\nu_i|=|\nu_1|+r$.
From Definition \ref{Def-adm}, the pair $(1,i)$ is a 
neighbourhood.
Similarly, when $\mathrm{sign}(\nu_1)=-$, there exists $i$
such that $|\nu_i|=|\nu_1|-r$. Then, the pair $(1,i)$ is 
a neighbourhood.
Thus if $\nu$ is admissible, then $\nu_1=J+(N-1)r$ or $\nu_{1}=-J$.

Let $\tilde{\nu}$ be an admissible element of length $N-1$ in 
$W_0\nu^{J,+}$. 
Suppose $\nu_1=J+(N-1)r$. There exists $i$ such that $|\nu_i|=J+(N-2)r$.
From Definition \ref{Def-adm}, the pair $(1,i)$ is not a 
neighbourhood. 
Set $\nu_1=J+(N-1)r$ and $\nu_i=\tilde{\nu}_{i-1}$ for $2\le i\le N$.
Since all the pairs $(j,k)$ with $2\le j<k\le N$ are not 
neighbourhood, $\nu$ is admissible.
Suppose $\nu_1=-J$. The pair $(1,i)$, $2\le i\le N$, is not 
a neighbourhood.   
Set $\nu_1=-J$ and 
$\nu_i=\mathrm{sign}(\tilde{\nu}_{i-1})(|\tilde{\nu}_{i-1}|+r)$ 
for $2\le i\le N$. 
Since the pairs $(i,j)$ with $2\le i<j\le N$ are not neighbourhood,
$\nu$ is admissible. 
From the construction of admissible elements, the dimension 
of $I^{(1,r+1)}(\nu^{J,+})$ is given by $2^{N}$.
\end{proof}

\begin{lemma}
\label{lemma-constraint}
Set $\nu^{\pm}:=\nu^{J,\pm}$. At $s^{2r}q^{4}=1$, we have 
\begin{eqnarray*}
(\hat{T}_i-q^{-1})E_{\nu^{\pm}}=0, \quad 1\le i\le N-1.
\end{eqnarray*}
\end{lemma}
\begin{proof}
The (anti-)dominant element $\nu^{\pm}$ is admissible and 
$s_i\cdot\nu^{\pm}$ $1\le i\le N-1$ is not admissible. 
There exist linear operators (called intertwiner) $\phi_i$, 
$0\le i\le N$ which send $E_{\lambda}$ to $E_{s_i\cdot\lambda}$, 
{\it i.e.}, $E_{s_i\cdot\lambda}=\phi_iE_{\lambda}$ 
(see Definition 2.4 in \cite{Kas08}). 
Especially, we have
\begin{eqnarray*}
\phi_i:=T_{i}+\frac{q^{-1}-q}{Y_{i+1}/Y_{i}-1}. 
\end{eqnarray*}
From Lemma 4.7 in \cite{Kas08}, we have $(\phi_iE_{\nu^{\pm}})|_{s^{2r}q^{4}=1}=0$ 
for $1\le i\le N-1$.
At $s^{2r}q^{4}=1$, the ratio of eigenvalues of $\hat{Y}$ is 
$y(\nu^{\pm})_{i+1}/y(\nu^{\pm})_{i}=q^{2}$.
Thus the action of $\phi_i$ on $E_{\nu^{\pm}}$ is equal to the 
action of $T_i-q^{-1}$.
\end{proof}

We consider a graph whose vertices are labelled by an admissible 
elements in $\mathbb{Z}^{N}$. 
We connect two vertices labelled by $\lambda$ and $\mu$ if 
and only if $\mu=s_{i}\lambda$ with $s_i\in W_0$. 
We put the integer $i$ on the edge connecting vertices labelled 
by $\lambda$ and $\mu$.
When an element $\lambda$ is admissible, 
we call this graph $\Gamma(\lambda)$. 

\begin{example}
Set $N=3$ and take $\nu^{1,+}$.  
We have eight admissible elements and the 
graph $\Gamma(\nu^{1,+})$ is as follows:
\begin{eqnarray*}
\begin{tikzpicture}
\path (0,0)node[shape=rectangle](a){$(3,2,1)$}
      (2.1,0)node[shape=rectangle](b){$(3,2,-1)$}
      (4.3,0)node[shape=rectangle](c){$(3,-1,2)$}
      (6.5,1)node[shape=rectangle](d){$(3,-1,-2)$}
      (6.5,-1)node[shape=rectangle](e){$(-1,3,2)$}
      (8.9,0)node[shape=rectangle](f){$(-1,3,-2)$}
      (11.4,0)node[shape=rectangle](g){$(-1,-2,3)$}
      (14,0)node[shape=rectangle](h){$(-1,-2,-3)$};
\draw(a.east)--(b.west)(b.east)--(c.west)
(c.east)--(d.west)(c.east)--(e.west)(d.east)--(f.west)
(e.east)--(f.west)(f.east)--(g.west)(g.east)--(h.west);
\draw(1,0)node[anchor=south]{$3$}
     (3.2,0)node[anchor=south]{$2$}
     (5.4,0.5)node[anchor=south east]{$3$}
     (5.4,-0.5)node[anchor=north east]{$1$}
     (7.7,0.5)node[anchor=south west]{$1$}
     (7.7,-0.5)node[anchor=north west]{$3$} 
     (10.15,0)node[anchor=south]{$2$} 
     (12.6,0)node[anchor=south]{$3$};
\end{tikzpicture}.
\end{eqnarray*}
\end{example}

\subsubsection{One-boundary case}
\begin{prop}
\label{prop-dim2}
The dimension of the space $I^{(2,2r+1)}(\xi^{0})$ is 
$\displaystyle2^{\lfloor N/2\rfloor}\genfrac{(}{)}{0pt}{}{N}{\lfloor N/2\rfloor}$. 
\end{prop}
\begin{proof}
Recall that the dominant element $\xi^{+}$ is $((N-1)r, (N-2)r, \ldots, r ,0)$ and this 
element is admissible and connected to $\xi^{0}$ and $\xi^{1}$. 
Suppose an element $\xi\in\mathbb{Z}_{\ge0}^{N}$ is admissible. 
From the definition of admissibility, $(N-i)r$ is left to $(N-i-2)r$ in $\xi$. 
The cardinality of $\xi$'s satisfying $\xi\in\mathbb{Z}_{\ge0}^{N}$ is 
$\genfrac{(}{)}{0pt}{}{N}{\lfloor N/2\rfloor}$. 
Take a sub-element $\xi'$ of length $\lfloor N/2\rfloor$ satisfying 
$\xi'\in W_{0}(2\lfloor N/2\rfloor-1,\ldots,3r,r)$. 
By a similar argument to Proposition~\ref{prop-dim}, the cardinality 
of admissible $\xi'$'s is $2^{\lfloor N/2\rfloor}$.
Thus the total dimension of $I^{(2,2r+1)}(\xi^{+})$ is 
$2^{\lfloor N/2\rfloor}\genfrac{(}{)}{0pt}{}{N}{\lfloor N/2\rfloor}$. 
\end{proof}

We denote by $I_{+}$ a vector space spanned by 
$\{E_{\lambda}| \lambda\in\mathbb{Z}_{\ge0}^{N}\}$. 
We define 
\begin{eqnarray*}
I_{+}^{(k,r')}=I^{(k,r')}\cap I_{+}.
\end{eqnarray*}
Let $\xi\in W_{0}\xi^{+}$ be an admissible element for $(k,r')=(2,2r+1)$ and $\xi_N=r$. 
We have intertwiners $\phi_{i}$, $0\le i\le N$, which send a non-symmetric 
Koornwinder polynomial to another one (see {\it e.g.} \cite[Definition 2.4]{Kas08}). 
We have $\phi_{N}E_{\xi}=c_{N,\xi}E_{s_{N}\cdot\xi}$. 
By a straightforward computation with Definition 2.4 in \cite{Kas08}, we have 
$c_{N,\xi}=0$ if $q_{N}^{2}=-q$. 
The space $I_{+}^{(2,2r+1)}$ is closed under the action of $Y_{i}$, $1\le i\le N$, 
and $T_{i}$, $0\le i\le N$, at $q_2^{2}=-q$.  
Thus, we have 
\begin{prop}
\label{prop-dim-I}
Suppose that $q_N^{2}=-q$. 
The space $I_{+}^{(2,2r+1)}$ is an irreducible representation of 
the affine Hecke algebra.
The dimension of $I_{+}^{(2,2r+1)}$ is 
$\genfrac{(}{)}{0pt}{}{N}{\lfloor N/2\rfloor}$.  
\end{prop}

\begin{prop}
At $s^{4r}q^{6}=1$ and $q_{N}^{2}=-q$, we have 
\begin{eqnarray*}
(\hat{T}_{i}-q^{-1})E_{\xi^{0}}&=&0,\qquad i\neq\lfloor(N+1)/2\rfloor, \\
(\hat{T}_{N}-q_{N}^{-1})E_{\xi^{0}}&=&0.
\end{eqnarray*}
\end{prop}
\begin{proof}
From the definition of the admissibility, an element 
$\xi':=s_{i}\cdot\xi^{0}$, $1\le i\le N-1$, $i\neq\lfloor(N+1)/2\rfloor$,  
is not admissible. 
Therefore, by a similar argument to the proof of Lemma~\ref{lemma-constraint},
we have $(\hat{T}_{i}-q^{-1})E_{\xi^{0}}=0$.

Let $\xi':=s_{N}\cdot\xi^{0}$. 
Since the element $\xi'$ is admissible, we have 
$\phi_{N}E_{\xi^{0}}=c_{N,\xi^{0}}E_{\xi'}$ with the intertwiner $\phi_{N}$ and
a rational function $c_{N,\xi^{0}}$ in $\mathbb{K}$ (see \cite[Proposition 2.5]{Kas08}). 
By a straightforward calculation, we have $c_{N,\xi^{0}}=0$ at $q_2^{2}=-q$. 
The intertwiner $\phi_{N}$ is written as 
\begin{eqnarray*}
\phi_{N}=\hat{T}_{N}+\frac{(q^{-1}_{N}-q_{N})+(q_0^{-1}-q_{0})Y_{N}^{-1}}{Y_{N}^{-2}-1}.
\end{eqnarray*}
Inserting $Y_N=q^{-1}q_0q_{N}$ at $s^{4r}q^{6}=1$ and $q_N^{2}=-q$, 
we obtain the action of $\phi_N$ is equal to the action of $\hat{T}_{N}-q_N^{-1}$ 
on $E_{\xi^{0}}$. 
This completes the proof.
\end{proof}

\section{Boundary quantum Knizhnik--Zamolodchikov equation}
\label{sec-qKZ}
We introduce boundary quantum Knizhnik-Zamolodchikov equations 
associated with the one- and two-boundary Temperley--Lieb algebras.

\subsection{Two-boundary case}
Let $\mathbf{z}:=(z_1,\ldots,z_N)$ and 
$\Psi(\mathbf{z})$ be a function taking values in 
$V^{\otimes N}$, {\it i.e.},
$\Psi(\mathbf{z}):=\sum_{b}\Psi_{b}(\mathbf{z})C_b$ 
where $C_b$ is the Kazhdan--Lusztig bases.
We define the scattering matrices $S_{i}, 1\le i\le N-1$,  
by 
\begin{eqnarray}
\label{def-scatter}
S_i(\mathbf{z})&:=&
\check{R}_{i}(z_i/(s^2z_{i+1}))\check{R}_{i+1}(z_i/(s^2z_{i+2}))
\ldots \check{R}_{N-1}(z_i/(s^2z_{N}))
K_{N}(s^2/z_i) \\
&&\times\check{R}_{N-1}(s^{-2}z_{i}z_N)\ldots\check{R}_{1}(s^{-2}z_{i}z_1)
K_{0}(s^{-1}z_{i}) \nonumber \\
&&\times\check{R}_{1}(z_i/z_1)\ldots\check{R}_{i-1}(z_i/z_{i-1})
\nonumber
\end{eqnarray}

The boundary quantum Knizhnik--Zamolodchikov equation \cite{Che91,Che92,FreRes92}
is 
\begin{eqnarray}
\label{qKZ-scatter}
S_i(\mathbf{z})\Psi(\mathbf{z})=
\Psi(z_1,\ldots,z_{i-1},s^{-2}z_{i},z_{i+1},\ldots,z_{N}), 
\quad 1\le i\le N.
\end{eqnarray}
Suppose that the function $\Psi$ satisfies 
\begin{eqnarray}
\label{qKZ-factor1}
\Psi(s_0\mathbf{z})&=&K_0(s^{-1}z_1)\Psi(\mathbf{z}), \\
\label{qKZ-factor2}
\Psi(s_i\mathbf{z})&=&\check{R}_{i}(z_{i+1}/z_{i})\Psi(\mathbf{z}), \\
\label{qKZ-factor3}
\Psi(s_N\mathbf{z})&=&K_{N}(z_{N})\Psi(\mathbf{z}).
\end{eqnarray}
Then, it is easy to show that $\Psi$ is a solution of the boundary qKZ 
equation. 
Hereafter, we solve the set of equations (\ref{qKZ-factor1}), 
(\ref{qKZ-factor2}) and (\ref{qKZ-factor3}) and 
call them boundary qKZ equation.

Define a linear operator $\hat{e}_{i}:=\hat{T}_{i}-q^{-1}$. 
The boundary quantum Knizhnik--Zamolodchikov equation is rewritten as 
\begin{eqnarray}
\label{qKZ-e}
e_i\Psi(\mathbf{z})=\hat{e}_{i}\Psi(\mathbf{z}), \quad
0\le i\le N.
\end{eqnarray}
Recall that $\Psi=\sum_{b}\Psi_{b}(\mathbf{z})C_b$. 
In Eqn.(\ref{qKZ-e}), the generator $e_i$ acts on a basis 
$C_b$ and the operator $\hat{e}_i$ acts on a function 
$\Psi_{b}(\mathbf{z})$.
We call Eqn.(\ref{qKZ-e}) with $1\le i\le N$ non-affine part of 
the boundary quantum Knizhnik--Zamolodchikov equation.

\subsection{One-boundary case}
Let $\Psi(\mathbf{z}):=\sum_{b\in\mathcal{B}_{N}}\Psi_{b}(\mathbf{z})C_b$ 
be a function taking values 
in $\mathcal{L}_{N}$. 
We define scattering matrices $\tilde{S}_{i}(\mathbf{z})$, $1\le i\le N$,  
by replacing $K_0(s^{-1}z_{1})$ with the identity in Eqn.(\ref{def-scatter}).
The boundary quantum Knizhnik--Zamolodchikov equation is of the form Enq.(\ref{qKZ-scatter}) by 
replacing $S_{i}(\mathbf{z})$ with $\tilde{S}_{i}(\mathbf{z})$. 
The factorized form of the boundary quantum Knizhnik--Zamolodchikov equation 
is Eqns.(\ref{qKZ-factor2}), (\ref{qKZ-factor3}) and 
\begin{eqnarray*}
\Psi(\mathbf{z})=\Psi(s_0\mathbf{z}).
\end{eqnarray*}

\section{Reduction of the boundary qKZ equation}
\label{sec-redqKZ}
\subsection{Two-boundary case}
Suppose that a binary string $\alpha\in\{+,-\}^{N}$ satisfies
$\alpha_{i}\ge\alpha_{i+1}$ for some $i$. 
There is no $C_{\beta}$ such that the expansion of $e_iC_{\beta}$
contains the term $C_{\alpha}$.
Therefore, we have $\hat{e}_{i}\Psi_{\alpha}(\mathbf{z})=0$.
Similarly, suppose that a binary string $\alpha$ satisfies 
$\alpha_{N}=+$. 
Then, there is no $C_{\beta}$ such that the expansion of 
$e_{N}C_{\beta}$ contains the term $C_{\alpha}$.
Thus, we have $\hat{e}_{N}\Psi_{\alpha}=0$.
We have non-trivial equations for $\hat{e}_i\Psi_{\alpha}$ 
when $\alpha_{i}<\alpha_{i+1}$ for $1\le i\le N-1$ or 
$\alpha_N=-$ for $i=N$.
Explicitly, non-trivial equations are written as 
\begin{eqnarray}
\label{qKZ-non-1}
(\hat{e}_{i}+q+q^{-1})\Psi_{\alpha}
=\Psi_{s_{i}\cdot\alpha}+
\sum_{\beta>\alpha}c^{i}_{\alpha\beta}\Psi_{\beta}, 
\quad 1\le i\le N-1,
\end{eqnarray}
where $(\alpha_{i},\alpha_{i+1})=(-,+)$, $\beta>\alpha$ is 
the lexicographic order, $\beta_i\ge\beta_{i+1}$ 
and $c_{\alpha\beta}^{i}\in\mathbb{K}$. 
Similarly, we also have 
\begin{eqnarray}
\label{qKZ-non-2}
(\hat{e}_{N}+q_{N}+q_{N}^{-1})\Psi_{\alpha}
=\Psi_{s_{N}\cdot\alpha}+
\sum_{\beta>\alpha}c^{N}_{\alpha\beta}\Psi_{\beta}, 
\end{eqnarray}
where $\alpha_{N}=-$, $\beta_{N}=+$ and $c^{N}_{\alpha\beta}\in\mathbb{K}$.

Set 
$b_j:=(\underbrace{-\ldots-}_{j-1}+\underbrace{-\ldots-}_{N-j})$ 
for $1\le i\le N$ and $b_0:=(-\ldots-)$ 
and $\Psi_{i}:=\Psi_{b_{i}}$.
\begin{lemma}
\label{lemma-red}
The non-affine part of quantum KZ equations 
$e_i\Psi(\mathbf{z})=\hat{e}_{i}\Psi(\mathbf{z})$, $1\le i\le N$
is equivalent to the following set of equations: 
Eqns.(\ref{qKZ-non-1}), (\ref{qKZ-non-2}) 
and 
\begin{eqnarray}
\label{qKZ-non-3}
\hat{e}_{i}\Psi_{0}=0,\quad 1\le i\le N-1.
\end{eqnarray}
\end{lemma}
\begin{proof}
We prove Proposition in the case of type BI. 
One can prove Proposition for other types by a similar argument.  
To prove Proposition, it is enough to show that 
one can obtain $\hat{e}_{i}\Psi_{\alpha}=0$ with 
$\alpha_{i}\ge\alpha_{i+1}$ for $1\le i\le N-1$ and 
with $\alpha_{N}=+$ for $i=N$ from Eqns.(\ref{qKZ-non-1}), 
(\ref{qKZ-non-2}) and (\ref{qKZ-non-3}).

We prove Proposition by induction in the reversed lexicographic 
order. 
Since $e_NC_{-\ldots-+}=\kappa^{-1}C_{-\ldots-}+\ldots$, 
we have  
\begin{eqnarray*}
\hat{e}_{i}\Psi_{N}&=&\hat{e}_{i}(\hat{e}_{N}+q_N+q_N^{-1})\Psi_{0} \\
&=&(\hat{e}_{N}+q_N+q_N^{-1})\hat{e}_{i}\Psi_{0} \\
&=&0
\end{eqnarray*}
where $1\le i\le N-2$ and we have used the commutation relation of $\hat{e}_{i}$.

Suppose that Proposition holds true up to $\beta>\alpha$. 
We have three cases for $(\alpha_{i},\alpha_{i+1})$: 1) $(+,-)$ 
2)$(-,-)$ and 3) $(+,+)$. We also have a case for $\alpha_N$: 
4) $\alpha_N=+$.
We consider only the case 1, 2 and 4 since one can apply a similar 
argument to case 3.

\paragraph{\bf Case 1}
Let $P_i$ and $Q_i$ be the statements:  
\begin{enumerate}
\item[($P_i$)] The $i$-th and $(i+1)$-th site are connected via a dashed arc,
\item[($Q_{i})$] The $i$-th arrow is the down arrow with a star.
\end{enumerate}
For a binary string $\epsilon\in\{\pm\}^{N}$, we define 
$\theta(R;\epsilon)=\Psi_{\epsilon}$ if the diagram for $\epsilon$ 
satisfies the statement $R$, and $\theta(R;\epsilon)=0$ otherwise.
Eqn.(\ref{qKZ-non-1}) is explicitly written as follows.
When $(\alpha_1,\alpha_2)=(+,-)$ (the first and second sites are 
underlined), we have 
\begin{eqnarray}
\label{eqn-non-1}
\Psi_{\underline{+-}+\ldots}
&=&(\hat{e}_1+[2])\Psi_{\underline{-+}+\ldots}
-\kappa_{N}^2\theta(P_2;\underline{--}-\ldots), \\
\label{eqn-non-2}
\Psi_{\underline{+-}-\ldots}
&=&(\hat{e}_1+[2])\Psi_{\underline{-+}-\ldots}
-\Psi_{\underline{--}+\ldots}
-\kappa_N\theta(Q_2;\underline{--}-\ldots). 
\end{eqnarray}
When $(\alpha_{N-1},\alpha_N)=(+,-)$ (the $(N-1)$-th 
and $N$-th sites are underlined), we have 
\begin{eqnarray}
\label{eqn-non-3}
\Psi_{\ldots+\underline{+-}}
&=&(\hat{e}_{N-1}+[2])\Psi_{\ldots+\underline{-+}}
-\Psi_{\ldots-\underline{++}}
-\kappa_{N}\theta(Q_N;\ldots+\underline{--}), \\
\label{eqn-non-4}
\Psi_{\ldots-\underline{+-}}
&=&(\hat{e}_{N-1}+[2])\Psi_{\ldots-\underline{-+}}
-\kappa_{N}\Psi_{\ldots-\underline{--}}.
\end{eqnarray}
In general, if $(\alpha_i,\alpha_{i+1})=(+-)$ 
(the $i$-th and $(i+1)$-th sites are underlined)  
we have 
\begin{eqnarray}
\label{eqn-non-5}
&&\Psi_{\ldots+\underline{+-}+\ldots}
=(\hat{e}_{i}+[2])\Psi_{\ldots+\underline{-+}+\ldots}
-\Psi_{\ldots-\underline{++}+\ldots}
-\kappa_{N}^{2}\theta(P_{i+1};\ldots+\underline{--}-\ldots), \\
\label{eqn-non-6}
&&\Psi_{\ldots+\underline{+-}-\ldots}
=(\hat{e}_i+[2])\Psi_{\ldots+\underline{-+}-\ldots}
-\Psi_{\ldots+\underline{--}+\ldots}
-\Psi_{\ldots-\underline{++}-\ldots} 
-\kappa_{N}\theta(Q_{i+1};\ldots+\underline{--}-\ldots), \\
\label{eqn-non-7}
&&\Psi_{\ldots-\underline{+-}+\ldots}
=(\hat{e}_{i}+[2])\Psi_{\ldots-\underline{-+}+\ldots}
-\kappa_{N}^2\theta(P_{i+1};\ldots-\underline{--}-\ldots), \\
\label{eqn-non-8}
&&\Psi_{\ldots-\underline{+-}-\ldots}
=(\hat{e}_i+[2])\Psi_{\ldots-\underline{-+}-\ldots}
-\Psi_{\ldots-\underline{--}+}
-\kappa_{N}\theta(Q_{i+1};\ldots-\underline{--}-\ldots).
\end{eqnarray}
From Eqns.(\ref{eqn-non-1}) to (\ref{eqn-non-8}), 
we have $\hat{e}_i\Psi_{\alpha}=0$ by 
$\hat{e}_i^2=-[2]\hat{e}_i$ and the induction assumption 
$\hat{e}_i\Psi_{\beta}=0$ for $\beta_i\ge\beta_{i+1}$.

\paragraph{\bf Case 2}
Let $h$ be the largest integer such that $\alpha_h=+$ 
and $\alpha_{j}=-$ for $h+1\le j\le i-1$. 
There are three cases for $h$: a) there is no such $h$, 
b) $h=1$ and c) $h\ge2$.

\paragraph{Case 2-a}
Since the binary string $\alpha$ satisfies $\alpha_j=-$
with $1\le j\le i+1$, we have 
\begin{eqnarray*}
\Psi_{\alpha}=\sum_{k\ge0}\sum_{i_1,i_2,\ldots,i_k}
c_{i_1,\ldots,i_k}
\hat{e}_{i_1}\ldots\hat{e}_{i_k}\Psi_{0} 
\end{eqnarray*}
where $c_{i_1,\ldots,i_k}\in\mathbb{K}$ and 
$i_l\ge i+2$. 
Since $[\hat{e}_i,\hat{e}_{i_l}]=0$ and $\hat{e}_{i}\Psi_{0}=0$, 
we have $\hat{e}_i\Psi_{\alpha}=0$.

\paragraph{Case 2-b}
We have three subcases for $i$: i) $i=2$, ii) $i=3$, 
and iii) $i\ge4$.
We consider the case i) and ii) since one can apply a 
similar argument to the case iii). 

\paragraph{Case 2-b-i}
From Eqn.(\ref{eqn-non-2}), (\ref{eqn-non-7}) and (\ref{eqn-non-8}), 
we obtain 
\begin{eqnarray*}
\hat{e}_2\Psi_{\alpha}
&=&\hat{e}_2(\hat{e}_1+[2])(\hat{e}_2+[2])\Psi_{--+\ldots}
-\hat{e}_{2}\Psi_{--+\ldots} \\
&=&(\hat{e}_2\hat{e}_1\hat{e}_2-\hat{e}_{2})\Psi_{--+\ldots}  \\
&=&(\hat{e}_1\hat{e}_2\hat{e}_1-\hat{e}_{1})\Psi_{--\ldots}  \\
&=&0,
\end{eqnarray*}
where we have used the induction assumption and the braid relation 
for $\hat{e}_i$.

\paragraph{Case 2-b-ii}
We have 
\begin{eqnarray*}
\hat{e}_3\Psi_{+-\underline{--}\ldots}
&=&\hat{e}_3(\hat{e}_1+[2])\Psi_{-+\underline{--}\ldots}
-\kappa_N\hat{e}_3\theta(Q_2;--\underline{--})
-\hat{e}_3\Psi_{--\underline{+-}\ldots} \\
&=&0,
\end{eqnarray*}
where we have used the induction assumption and 
the commutation relations for $\hat{e}_i$.

\paragraph{Case 2-c}
We have two cases for $\alpha_{h-1}$: i) $\alpha_{h-1}=+$, 
and ii) $\alpha_{h-1}=-$. 
We consider the case i) only since one can apply a similar 
argument to the case ii.

\paragraph{Case 2-c-i}
We have three cases for $i$: $i=h+1$, $i=h+2$ and 
$i\ge h+3$. 
We consider the $i=h+1$ case since one can apply a similar 
argument to other cases.
\begin{eqnarray*}
\hat{e}_{i}\Psi_{\alpha}
&=&\hat{e}_{i}(\hat{e}_{i-1}+[2])\Psi_{+-\underline{+-}}
-\hat{e}_{i}\Psi_{\ldots+-\underline{-+}} \\
&=&\hat{e}_{i}(\hat{e}_{i-1}+[2])(\hat{e}_{i}+[2])\Psi_{+-\underline{-+}}
-\hat{e}_{i}\Psi_{\ldots+-\underline{-+}} \\
&=&0,
\end{eqnarray*}
where we have used the induction assumption and the braid 
relation for $\hat{e}_{i}$.

\paragraph{\bf Case 4}
Since we have 
$\hat{e}_N\Psi_{N}=\hat{e}_N(\hat{e}_N+q_N+q_N^{-1})\Psi_{0}=0$, 
Proposition holds true for $b_{N}$. 
We assume that Proposition is true up to $\beta$ with $\beta>\alpha$.
Since $\alpha_N=+$, from Eqn.(\ref{qKZ-non-2}), we have 
\begin{eqnarray*}
\hat{e}_N\Psi_{\alpha}
&=&\hat{e}_N(\hat{e}_N+q_N+q_N^{-1})\Psi_{s_N\cdot\alpha}
-\sum_{\beta>\alpha}c_{\alpha\beta}^{N}\hat{e}_N\Psi_{\beta} \\
&=&0,
\end{eqnarray*}
where we have used the induction assumption $\hat{e}_N\Psi_{\beta}=0$.
\end{proof}

We consider a graph whose vertices are labelled by a binary string 
in $\{\pm\}^{N}$. 
We connect two vertices labelled by $\alpha$ and $\beta$ 
if and only if $\alpha=s_{i}\beta$, $s_{i}\in W_0$ and 
$\alpha<\beta$ in the lexicographic order. 
Further, we put the integer $i$ on the edge connecting the vertices 
$\alpha$ and $\beta$.
We denote this graph by $\Gamma'$.
\begin{prop}
\label{prop-numofqKZ}
Let $D_{N}$ be the number of edges in $\Gamma'$. 
Then, $D_N$ satisfies the following recurrence 
relation: 
\begin{eqnarray*}
D_N=2D_{N-1}+2^{N-2}
\end{eqnarray*} 
with $D_2=3$.
\end{prop}
\begin{proof}
From the construction of $\Gamma'$, the number of edges 
in $\Gamma'$ is equal to the number of non-trivial 
equations (\ref{qKZ-non-1}) and (\ref{qKZ-non-2}).
Let $d_{\alpha}$ be the number of $i$, $1\le i\le N$,  
such that $(\alpha_{i},\alpha_{i+1})=(-,+)$ or 
$\alpha_N=-$.
Then we have $D_N=\sum_{\alpha}d_{\alpha}$. 
The number of $\alpha$ with $\alpha_N=-$ is 
$2^{N-1}$.
When $\alpha_{N}=-$, the partial sum 
$\sum_{\alpha,\alpha_N=-}d_{\alpha}$
is equal to $2^{N-1}+(D_{N-1}-2^{N-2})$.
When $\alpha_{N}=+$, the partial sum 
$\sum_{\alpha,\alpha_N=+}d_{\alpha}$ 
is equal to $D_{N-1}$. 
Therefore, we have 
\begin{eqnarray*}
D_N&=&2^{N-1}+D_{N-1}+(D_{N-1}-2^{N-2}) \\
&=&2D_{N-1}+2^{N-2}.
\end{eqnarray*}
By a direct computation, we have $D_{2}=3$. 
This completes the proof.
\end{proof}
By a straightforward computation, one can show that
\begin{cor}
We have $D_{N}\ge 2^{N}-1$. 
\end{cor}

We define a map $\varphi_{\pm}$ from an admissible element $\nu$ in 
$W_0\nu^{J,+}$ to  a binary string $\alpha$ of length $N$:   
\begin{eqnarray*}
\varphi_{\pm}:\nu_i\mapsto\alpha_i=\mathrm{sign}(\pm\nu_i), \quad 1\le i\le N.
\end{eqnarray*}
The inverse $\varphi_{\pm}^{-1}$ is obtained by the following algorithm:
\begin{eqnarray*}
\varphi_{+}: 
\alpha_i\mapsto\nu_{i}&=&
\begin{cases}
\max(S_{i}), & \text{for\ } \alpha_i=+, \\
-\min(S_{i}), & \text{for\ } \alpha_i=-, 
\end{cases} \\
\varphi_{-}: 
\alpha_i\mapsto\nu_{i}&=&
\begin{cases}
-\min(S_{i}), & \text{for\ } \alpha_i=+, \\
\max(S_{i}), & \text{for\ } \alpha_i=-, 
\end{cases} \\
S_{i+1}&=&S_{i}\setminus{|\nu_{i}|}
\end{eqnarray*}
with $S_{1}:=\{J,J+r,\ldots,J+(N-1)r\}$.
From the explicit construction, $\varphi_{\pm}$ is a bijection. 

\begin{prop}
\label{prop-diagram}
We have $\Gamma(\nu^{J,+})=\Gamma'$ as a graph.
\end{prop}
\begin{proof}
The number of vertices in $\Gamma(\nu^{J,+})$ and $\Gamma'$ 
is $2^{N}$. 
We have a bijection $\varphi_{\pm}$ from vertices in 
$\Gamma(\nu^{J,+})$ to vertices in $\Gamma'$. 
To prove $\Gamma(\nu^{J,+})=\Gamma'$, it is enough to 
show that vertices $\mu$ and $\nu$ are connected by 
an edge with an integer $i$ in $\Gamma(\nu^{J,+})$  
if and only if vertices $\alpha=\varphi_{\pm}(\mu)$
and $\beta=\varphi_{\pm}(\nu)$ are connected by an edge 
with the integer $i$ in $\Gamma'$.
We prove Proposition for $\varphi_{+}$ since one can apply 
the same argument to $\varphi_{-}$.

We first show that if the vertex $\alpha=\varphi_{+}(\mu)$ 
does not have an edge with integer $i$ in $\Gamma'$, then 
$s_i\mu$ is not admissible. 
We have $\alpha_{i}=\alpha_{i+1}$. 
From the explicit construction of $\varphi_{\pm}^{-1}$, 
we have $\mu_{i+1}=\mu_{i}-1$ for $\alpha_i=+$ 
and $\mu_{i+1}=\mu_{i}+1$ for $\alpha_i=-$. 
From Definition \ref{Def-adm}, the pair $(i,i+1)$ in 
$s_i\mu$ is a neighbourhood. 
Thus $s_{i}\mu$ is not admissible.

Below, we will show that if $\Gamma'$ has an edge with the integer 
$i$, then $\Gamma(\nu^{J,+})$ has the corresponding edge.
Suppose that $\alpha=s_i\beta$ with $\alpha<\beta$.
Then, in $\Gamma'$, the two vertices $\alpha$ and 
$\beta$ is connected and the edge has the integer $i$.
We have two cases for $i$: 1) $1\le i\le N-1$, and 
2) $i=N$. 

\paragraph{Case 1}
Let $\nu=\varphi_{+}(\beta)$. 
Since $(\beta_{i},\beta_{i+1})=(-,+)$, 
we have $|\nu_{i+1}|-|\nu_{i}|\ge2$ for $1\le i\le N-2$.
In this case, $s_i\nu$ is admissible since $s_i\nu$ does 
not have a neighbourhood.  
Suppose that $i=N-1$. 
We have $(\nu_{N-1},\nu_{N})=(-p,p+1)$ for some $p\in\mathbb{Z}_{>0}$.
Let $\mu=s_{i}\nu$. 
We have $|\rho(\mu)_{N-1}|-|\rho(\mu)_{N}|=1$, 
$|\mu_{N-1}|-|\mu_{N}|=1$ and $(\sigma(\mu)_{N-1},\sigma(\mu)_{N})=(+,-)$. 
From the definition of neieghbourhood (see Definition \ref{Def-adm}), 
$\mu$ is admissible. Thus we have an edge with the integer $i$.

\paragraph{Case 2}
Let $\nu=\varphi_{+}(\beta)$ and $\mu=s_{i}\nu$. 
Since $\beta_N=-$, we have $\nu_{N}=-p$ with some $p\in\mathbb{Z}_{>0}$.
From the explicit construction of $\varphi_{+}$, there exists no
integer $i$, $1\le i\le N-1$, such that $v_{i}=-(p+1)$ or $v_{i}=(p-1)$. 
We have two cases for $v_i$ for some $1\le i\le N-1$: 
a) $v_i=p+1$ and b) $v_i=-(p-1)$.
We consider only case a since we can apply the essentially same 
argument to case b.

\paragraph{Case 2-a}
We have $|\rho(\mu)_i|-|\rho(\mu)_N|=1$, 
$|\mu_i|-|\mu_{N}|=1$ and $(\sigma(\mu)_i,\sigma(\mu)_N)=(+,+)$. 
From Definition \ref{Def-adm}, $(i,N)$ is not a neighbourhood. 
Thus $\mu$ is admissible. 
In $\Gamma(\nu^{J,+})$,  vertices $\mu$ and $\nu$ are connected 
by an edge with the integer $i$.
This completes the proof.
\end{proof}

From Proposition~\ref{prop-numofqKZ}, we have $D_N$ non-trivial 
equations.
Note that if $\alpha=s_i\beta$ with $\alpha<\beta$, we have 
a non-trivial equation (\ref{qKZ-non-1}) or (\ref{qKZ-non-2}).
An edge of the graph $\Gamma'$ encodes these non-trivial equations. 
Since the graph $\Gamma'$ is connected, for any $v\in\{\pm\}^{N}$, 
there exists a sequence of vertices 
$\mathbf{v}:=(v_0,v_1,\ldots,v_l)$ 
and a sequence of integers $\mathbf{i}:=(i_1,\ldots,i_l)$ such that 
$v_{j-1}>v_{j}$ and the vertices $v_{j-1}$ and $v_j$ are connected by an edge 
with an integer $i_j$. 
We call the doublet $(\mathbf{v,i})$ a path from $v_0$ to 
$v_l$.

Suppose that $p=(\mathbf{v,i})$ with $\mathbf{v}=(v_0,v_1,v_2)$ and 
$\mathbf{i}=(p,q)$ is a path from $v_0$ to $v_2$ with $|p-q|>1$. 
Since $v_2=s_qs_pv_0=s_ps_qv_0$, we have another path 
$p'=(\mathbf{v',i'})$ with $\mathbf{v}'=(v_0,v'_1,v_2)$ and $\mathbf{i}'=(q,p)$. 
Although We have two non-trivial equations (\ref{qKZ-non-1}) 
corresponding to paths, these two equations are compatible 
with $\hat{e}_{p}\hat{e}_q=\hat{e}_{q}\hat{e}_{p}$ and
$e_{p}e_q=e_{q}e_{p}$.  
By a local change of a path, we mean that we change a partial 
path $p$ to $p'$. 
Notice that there is no partial path $(\mathbf{v,i})$ with 
$\mathbf{v}=(v_0,v_1,v_2,v_3)$ and $\mathbf{i}=(p,p+1,p)$ 
in $\Gamma'$. 

Suppose that there are several paths from $b_0$ to $v$. 
Fix a path $p_{v}=(\mathbf{v,i})$ form $b_0$ to $v$. 
One can obtain all the other paths from $p$ by successive 
local changes of $p$.
Therefore, we have one non-trivial equation which is 
associated with the path $p_v$. 
We obtain
\begin{prop}
\label{prop-redqKZ}
The non-affine part of boundary qKZ equation is equivalent 
to the following set of $2^{N}+N-2$ equations:
\begin{enumerate}
\item $\hat{e}_i\Psi_0=0, 1\le i\le N-1$. 
\item Given $v<b_0$, fix a path $p_v:=(\mathbf{v,i})$ with 
$\mathbf{i}=(i_1,\ldots,i_l)$. Then, 
\begin{eqnarray}
\label{redqKZ}
e_{i_l}\ldots e_{i_1}\Psi(\mathbf{z})
=\hat{e}_{i_l}\ldots \hat{e}_{i_1}\Psi(\mathbf{z}). 
\end{eqnarray}
\end{enumerate}

\end{prop}

Let $W_{i}^{p}$ be the following statement: 
\begin{center}
($W_{i}^{p}$) The $i$-th arrow is the down arrow
with the integer $p$.
\end{center}
For a binary string $\epsilon\in\{\pm\}^{N}$, we define 
$\theta(W_{i}^{p};\epsilon):=\Psi_{\epsilon}$ if the 
diagram for $\epsilon$ satisfies the statement $W_{i}^{p}$,
and $\theta(W_{i}^{p};\epsilon):=0$ otherwise.
Set $\widetilde{\Psi}_{i}:=\Psi_{i}-q^{-1}\Psi_{i+1}$ 
for $1\le i\le N-1$ and $\widetilde{\Psi}_{N}:=\Psi_{N}$. 

\begin{lemma}
\label{lemma-act-0}
For type BI, we have 
\begin{eqnarray}
\label{act-0-0}
&&\hat{T}_{0}\Psi_{0}=\kappa_{0}\widetilde{\Psi}_{1} 
-\kappa_{0}\kappa_{N}q^{-p}\theta(W_{1}^{p};b_0), \\
\label{act-0-1}
&&\hat{T}_{i}\widetilde{\Psi}_{i}
=\widetilde{\Psi}_{i+1}
-\kappa_{N}q^{-1}\theta(W_{i+1}^{1};b_{0}), \quad 1\le i\le N-1,  \\
\label{act-0-3}
&&\hat{T}_{N}(\Psi_{N}-\kappa_{N}q^{-M}\Psi_{0})
=\kappa_{N}\Psi_{0}.
\end{eqnarray}
\end{lemma}
\begin{proof}
We show Eqn.(\ref{act-0-0}) since one can apply a similar 
argument to other cases.
We have $e_0(C_{b_1})=\kappa_{0}C_{b_0}+\ldots$ and 
$e_0(C_{b_{2}}	)=-\kappa_{0}q^{-1}C_{b_{0}}+\ldots$. 
When the diagram for $b_0$ satisfies the statement $W_{1}^{p}$, 
we have $e_0C_{b_0}=-\kappa_0\kappa_Nq^{-p}C_{b_0}+\ldots$.
This implies Eqn.(\ref{act-0-0}).
\end{proof}

\begin{lemma}
\label{lemma-act-1}
For type BII, we have 
\begin{eqnarray*}
&&\hat{T}_0\Psi_{0}
=\kappa_0\widetilde{\Psi}_{1}
+
\begin{cases}
\kappa_{N}\kappa_{0}q^{-1}q_{N}\Psi_{0}, & \text{for $N$ even}, \\
-\kappa_{N}\kappa_{0}q_{N}^{-1}\Psi_{0}, & \text{for $N$ odd},
\end{cases}\\
&&\hat{T}_{i}\widetilde{\Psi}_{i}
=\widetilde{\Psi}_{i+1}+
\begin{cases}
-\kappa_{N}q^{-1}(q^{-1}q_{N}+qq_{N}^{-1})\Psi_{0}, 
& \text{for } i\equiv N+1 (\mathrm{mod}\ 2), \\
\kappa_{N}q^{-1}(q_{N}+q_{N}^{-1})\Psi_{0}, 
& \text{for } i\equiv N (\mathrm{mod}\ 2), 
\end{cases} \\ 
&&\hat{T}_{N}(\widetilde{\Psi}_{N}-\kappa_{N}q_{N}^{-1}\Psi_{0})
=\kappa_{N}\Psi_{0}.
\end{eqnarray*}
\end{lemma}

\begin{lemma}
\label{lemma-act-2}
For type BIII, we have
\begin{eqnarray*}
&&\hat{T}_{0}\Psi_{0}
=\kappa_0(\widetilde{\Psi}_{1}-\kappa_{N}q^{N-1}q_{N}^{-1}\Psi_{0}), \\
&&\hat{T}_{i}\widetilde{\Psi}_{i}=\widetilde{\Psi}_{i+1}, \qquad 1\le i\le N-1,\\
&&\hat{T}_{N}(\widetilde{\Psi}_{N}-\kappa_{N}q_{N}^{-1}\Psi_{0})
=\kappa_{N}\Psi_{0}.
\end{eqnarray*}
\end{lemma}

\begin{prop}
\label{prop-Y-0}
We have 
\begin{eqnarray*}
\hat{Y}_{i}\Psi_{0}=\kappa_{0}\kappa_{N}q^{-N+2i-1}\Psi_{0}.
\end{eqnarray*}
\end{prop}
\begin{proof}
Since there is no basis $C_b$ such that the expansion of $e_i(C_b)$ contains
the basis $C_{b_0}$, we have $\hat{e}_{i}\Psi_{0}=0$ for $1\le i\le N-1$. 
Thus we have $\hat{T}_{i}^{\pm1}\Psi_{0}=q^{\mp1}\Psi_{0}$.
From the correspondence (\ref{BZ-TY}), we have 
\begin{eqnarray*}
\hat{Y}_{i}\Psi_{0}
&=&q^{i-1}\hat{T}_i\ldots\hat{T}_{N-1}\hat{T}_{N}\ldots\hat{T}_{0}\Psi_{0} \\
&=&\kappa_{0}\kappa_{N}q^{-N+2i-1}\Psi_{0}
\end{eqnarray*}
where we have successively used Lemma~\ref{lemma-act-0}, \ref{lemma-act-1} or 
\ref{lemma-act-2}.
\end{proof}

\subsection{One-boundary case}
Let $b_0, b_r\in\{\pm\}^{N}$ be binary strings as
\begin{eqnarray*}
&&b_0:=(\underbrace{-\ldots-}_{N/2}\underbrace{+\ldots+}_{N/2}), 
\text{ for $N$ even}, \\
&&b_{0}:=(\underbrace{-\ldots-}_{(N+1)/2}\underbrace{+\ldots+}_{(N-1)/2}), \quad
b_{r}:=(\underbrace{-\ldots-}_{(N-1)/2}\underbrace{+\ldots+}_{(N-1)/2}-),
\text{ for $N$ odd}.
\end{eqnarray*}
We abbreviate $\Psi_{b_0}$ and $\Psi_{b_r}$ as $\Psi_0$ and $\Psi_{r}$.

We show that one can obtain a component of $\Psi(\mathbf{z})$ from 
the generating vector $\Psi_{0}$ through the boundary qKZ equation.

Given two binary strings $b_1:=(b_{1,i}\ldots b_{1,N})\in\{\pm\}^{N}$ 
and $b_2=(b_{2,1}\ldots b_{2,N})\in\{\pm\}^{N}$, we denote 
by $b_1\preceq b_2$ if and only if $\sum_{i=1}^{j}b_{1,i}\le\sum_{i=1}^{j}b_{1,i}$
for $1\le j\le N$. 
Set $b_{-}:=\{-\}^{N}$ and abbreviate $\Psi_{b_{-}}$ as $\Psi_{-}$. 
The binary string $b_-$ satisfies $b_{-}\prec b$ for $b\in\mathcal{B}_{N}$.

We define $\mathcal{B}_{N}^{(i)}:=\{b\in\mathcal{B}_{N}|\sum_{j=1}^{N}b_{j}=-i\}$ 
for $0\le i\le N$.  
Let $\mathcal{L}^{(i)}_{N}$, $0\le i\le N$, be a subspace in $\mathcal{L}_{N}$ 
spanned by $\{C_{b}| b\in\mathcal{B}^{(i)}_{N}\}$. 
The dimension of the space $\mathcal{L}^{(i)}_{N}$ is zero if $N-i\equiv1\pmod{2}$
and one if $i=N$.

\begin{lemma}
\label{lemma-red-qKZ-1b}
Given the component $\Psi_{-}$, a component $\Psi_{b}$ with 
$b\in\mathcal{B}_{N}$ is written in terms of $\Psi_{-}$ 
through the boundary quantum Knizhnik--Zamolodchikov equation.
\end{lemma}
\begin{proof}
We prove Lemma by induction.
Let $b_{i}=(b_{i,1}\ldots b_{i,N})\in\mathcal{B}^{(N-2)}_{N}$, $1\le i\le N-1$ 
be a binary string such that $(b_{i,i},b_{i,i+1})=(-,+)$ and $b_{i,j}=-$ 
for $j\neq i,i+1$.
Then, we have $b_{-}\prec b_{N-1}\prec\cdots\prec b_{1}$. 
We have $e_{N}(C_{b_{N-1}})=\kappa_{N}^{-1}C_{b_{-}}$ and there exists no 
$b\neq b_{N-1}$ such that $C_{b_{-}}$ appears in the expansion of $e_{N}(C_{b})$. 
The boundary qKZ equation for $b_{-}$-component is written as 
\begin{eqnarray}
\label{eq-Psim1}
\Psi_{b_{N-1}}=\kappa_{N}(\hat{T}_{N}+q_{N})\Psi_{b_{-}}.
\end{eqnarray}
Similarly, we have $e_{i}(C_{b_{i-1}})=e_{i}(C_{b_{i+1}})=C_{b_{i}}$ and 
$e_{i}(C_{b_{-}})=\alpha C_{b_{i}}$ where $\alpha$ is 
$\kappa_{N}(qq^{-1}_N+q^{-1}q_{N})$ (resp. $-\kappa_{N}(q_{N}+q_{N}^{-1})$) 
for $N-i$ odd (resp. even). 
The boundary qKZ equation is written as 
\begin{eqnarray}
\label{eq-Psim2}
\Psi_{b_{i-1}}=(\hat{T}_{i}+q)\Psi_{b_{i}}-\Psi_{b_{i+1}}-\alpha\Psi_{-},
\end{eqnarray}
for $1\le i\le N-1$ and $\Psi_{b_{N}}=0$. 
From Eqns.(\ref{eq-Psim1}) and (\ref{eq-Psim2}), the component 
$\Psi_{b_{i}}$, $1\le i\le N-1$ is written in terms of $\Psi_{b_{-}}$.  

Fix a binary string $b\in\mathcal{B}^{(p)}_{N}$ with $p\neq N$. 
We assume that Lemma holds true for all $b'\preceq b$. 
Since $b\neq b_{-}$, there exists $i$ such that $(b_{i},b_{i+1})=(-,+)$.
Recall a diagram for the Kazhdan--Lusztig basis indexed by 
the binary string $b$. 
The diagram has a little arc $a$ connecting the $i$-th and the $(i+1)$-th 
sites. 
We have two cases for $a$: 1) there exists a larger arc outside of $a$, 
and 2) there exists no larger arc outside of $a$.

\paragraph{Case 1}
Let $b'$ be a partial string $b':=(b_{i-1}b_{i}b_{i+1}b_{i+2})$. 
We have four cases for $b'$. 
When $b'=(+-++)$, set a partial string $b_{1}:=(-+++)$ or $(++-+)$. 
When $b'=(+-+-)$, set a partial string $b_{1}:=(-++-)$, $(+--+)$ or $(++--)$.
When $b'=(--++)$, set a partial string $b_{1}:=(-+-+)$.
When $b'=(--+-)$, set a partial string $b_{1}:=(-+--)$ or $(---+)$.
Then, we have $e_{i}(C_{b_{1}})=C_{b'}$.
Note that there exists at most one partial string $b_1$ such that 
$b'\prec b_{1}$. 
In the first case, the boundary qKZ equation for $b$-component is written as 
\begin{eqnarray*}
\Psi_{++-+}=(\hat{T}_{i}+q)\Psi_{b}-\Psi_{-+++}.
\end{eqnarray*}
Since the component $\Psi_{-+++}$ is written in terms of 
$\Psi_{-}$ from the induction assumption, $\Psi_{++-+}$ is also 
written in terms of $\Psi_{-}$.
One can prove Lemma for other cases by a similar argument.

\paragraph{Case 2}
Let $b'$ be a partial binary string $b'=(b_{i-1}b_{i}b_{i+1}b_{i+2})$. 
We have two cases: i) $b'=(+-+-)$ and ii) $b'=(--+-)$.

In the case of i), we assume that there exists $o$- or $e$-unpaired 
down arrow at the $j$-th ($j\le i-1$) site. 
Let $b_{1}=(++--), (+--+)$ or $(-++-)$ and 
$b_{2}=(+---)\in\mathcal{B}^{(p+2)}_{N}$. 
Then, we have $e_{i}(C_{b_{1}})=C_{b}$, $e_{i}(C_{b})=-(q+q^{-1})C_{b}$ 
and $e_{i}(C_{b_{2}})=\kappa_{N}^{-1}C_{b}$.  
Note that $(++--)$ satisfies $b\prec(++--)$ and other binary strings are 
smaller than $b$ with respect to $\prec$. 
We have 
\begin{eqnarray*}
\Psi_{++--}=(\hat{T}_{i}+q)\Psi_{b}-\Psi_{+--+}-\Psi_{-++-}
-\kappa_{N}^{-1}\Psi_{+---}.
\end{eqnarray*}
Since the right hand side of the above equation is written in terms of $\Psi_{-}$ 
from the induction assumption, $\Psi_{++--}$ is written in terms of $\Psi_{-}$. 
We can prove Lemma  for Case ii) by a similar argument to Case i). 
This completes the proof.
\end{proof}

\begin{lemma}
\label{lemma-T-1b}
We have 
\begin{eqnarray}
\label{cond-1b-1}
(\hat{T}_{i}-q^{-1})\Psi_{0}&=&0, \quad 1\le i\le N-1, i\neq \lfloor(N+1)/2\rfloor, \\
\label{cond-1b-2}
(\hat{T}_{N}-q_{N}^{-1})\Psi_{0}&=&0.
\end{eqnarray}
\end{lemma}
\begin{proof}
There is no $b\in\mathcal{B}_{N}$ such that $C_{b_{0}}$ appears 
in the expansion of $e_{i}C_{b}$ for $i\neq\lfloor(N+1)/2\rfloor$.
The $b_{0}$-component of $\Psi$ satisfies Eqns.(\ref{cond-1b-1}) and 
(\ref{cond-1b-2}).
\end{proof}

First consider $N$ even. 
Given a binary string 
$b_i=(\underbrace{-\ldots-}_{N/2-1}\underbrace{+\ldots+}_{i}-\underbrace{+\ldots+}_{N/2-i})$,
$1\le i\le N/2-1$, we abbreviate $\Psi_{i}:=\Psi_{b_i}$. 
For a binary string $b_{N/2}:=(\underbrace{-\ldots-}_{N/2-1}\underbrace{+\ldots+}_{N/2-1}--)$, 
we abbreviate $\Psi_{N/2}:=\Psi_{b_{N/2}}$.
Set $\tilde{\Psi}^{\pm}_{i}:=\Psi_{i}-q^{\pm1}\Psi_{i-1}$ for $0\le i\le N/2-1$ 
with $\Psi_{-1}=0$.

\begin{lemma}
\label{lemma-one-act-1}
We have 
\begin{eqnarray*}
\hat{T}_{N/2+i}\tilde{\Psi}^{+}_{i}&=&\tilde{\Psi}^{+}_{i+1}, \qquad 
\text{for } 0\le i\le N/2-2, \\
\hat{T}_{N-1}\tilde{\Psi}^{+}_{N/2-1}&=&
-q\Psi_{N/2-1}+\alpha\Psi_{N/2}, \\
\hat{T}_{N}(-q\Psi_{N/2-1}+\alpha\Psi_{N/2})&=&
q_N(q^{-1}\Psi_{N/2-1}-\alpha\Psi_{N/2}), \\
\hat{T}_{N-1}(q^{-1}\Psi_{N/2-1}-\alpha\Psi_{N/2})
&=&-\tilde{\Psi}^{-}_{N/2-1}, \\
\hat{T}_{N/2+i}\tilde{\Psi}^{-}_{i+1}&=&
\tilde{\Psi}^{-}_{i}, \quad\text{for  } 0\le i\le N/2-2,
\end{eqnarray*}
where $\alpha=\kappa_{N}(qq_{N}^{-1}+q^{-1}q_{N})$. 
\end{lemma}
\begin{proof}
Since we have $e_{N/2+i}(C_{b_{i+1}})=C_{b_{i}}$, $e_{N/2+i}(C_{b_{i-1}})=C_{b_{i}}$
and $e_{N/2+i}(C_{b_{i}})=-(q+q^{-1})C_{b_{i}}$, 
we have 
\begin{eqnarray*}
\hat{T}_{N/2+i}\tilde{\Psi}^{+}_{i}&=&q^{-1}\Psi_{i}-\Psi_{i-1}-(q+q^{-1})\Psi_{i} 
+\Psi_{i-1}+\Psi_{i+1} \\
&=&\tilde{\Psi}^{+}_{i+1},
\end{eqnarray*}
for $0\le i\le N/2-2$. 
Other equations can be proven in a similar way.
\end{proof}

Secondly, we consider $N$ odd.
Given a binary string 
$b_i:=(\underbrace{-\ldots-}_{(N-1)/2}\underbrace{+\ldots+}_{i}-
\underbrace{+\ldots+}_{(N-1)/2-i})$, $0\le i\le (N-3)/2$, 
we abbreviate $\Psi_{i}:=\Psi_{b_i}$. 
Let $b_{(N-1)/2}$ be a binary string 
$b_{(N-1)/2}:=(\underbrace{-\ldots-}_{(N-1)/2}\underbrace{+\ldots+}_{(N-3)/2}--)$
and $b_{(N+1)/2}:=(\underbrace{-\ldots-}_{(N-1)/2}\underbrace{+\ldots+}_{(N-1)/2}-)$. 
We abbreviate $\Psi_{(N\pm1)/2}:=\Psi_{b_{(N\pm1)/2}}$.
We set $\tilde{\Psi}^{\pm}_{i}:=\Psi_{i}-q^{\pm1}\Psi_{i-1}$ for $0\le i\le(N-3)/2$ 
and $\Psi_{-1}=0$.
\begin{prop}
\label{lemma-one-act-2}
We have 
\begin{eqnarray*}
\hat{T}_{(N+1)/2+i}\tilde{\Psi}_{i}^{+}&=&\tilde{\Psi}_{i+1}^{+}, \qquad 
\text{for } 0\le i\le (N-5)/2, \\
\hat{T}_{N-1}\tilde{\Psi}_{(N-3)/2}^{+}&=&
-q\Psi_{(N-3)/2}+\alpha\Psi_{(N-1)/2}, \\
\hat{T}_{N}(-q\Psi_{(N-3)/2}+\alpha\Psi_{(N-1)/2}) &=&
q_{N}(q^{-1}\Psi_{(N-3)/2}-\alpha\Psi_{(N-1)/2}), \\
\hat{T}_{N-1}(q^{-1}\Psi_{(N-3)/2}-\alpha\Psi_{(N-1)/2})&=&
-\tilde{\Psi}^{-}_{(N-3)/2}, \\
\hat{T}_{(N+1)/2+i}\tilde{\Psi}^{-}_{i+1}&=&
\tilde{\Psi}^{-}_{i}, \qquad \text{for } 0\le i\le (N-3)/2
\end{eqnarray*}
where $\alpha=\kappa_{N}(qq^{-1}_{N}+q^{-1}q_{N})$. 
\end{prop}
We omit a proof of Proposition~\ref{lemma-one-act-2} since 
one can apply a similar argument in a proof of Proposition~\ref{lemma-one-act-1}  
to this case.

\begin{prop}
\label{prop-Y-0-2}
We have 
\begin{eqnarray*}
\text{N: odd \ }\quad
&&\hat{Y}_{i}\Psi_{0}
=
\begin{cases}
q^{-(N+1-2i)}q_{0}q_N\Psi_{0}, & 1\le i\le (N+1)/2, \\
-q^{-2(N-i)}q_{0}q_{N}^{-1}\Psi_{0}, & (N+3)/2\le i\le N,
\end{cases} \\
&&\hat{Y}_{i}\Psi_{(N+1)/2}
=
\begin{cases}
-q^{-(N-1-2i)}q_{0}q_N^{-1}\Psi_{(N+1)/2}, & 1\le i\le (N-1)/2, \\
q^{-2(N-i)}q_{0}q_N\Psi_{(N+1)/2}, & (N+1)/2\le i\le N, 
\end{cases} \\
\text{N: even }\quad
&&\hat{Y}_{i}\Psi_{0}
=
\begin{cases}
q^{-(N-2i)}q_{0}q_N\Psi_{0}, & 1\le i\le N/2, \\
-q^{-2(N-i)}q_{0}q_N^{-1}\Psi_{0}, & N/2\le i\le N.
\end{cases}
\end{eqnarray*}
\end{prop}
\begin{proof}
By a straightforward calculation, one can verify Proposition holds true
for $N=2, 3$.
Consider N even first. 
We have $\hat{T}_{i}\Psi_{0}=q^{-1}\Psi_{0}$ for $i\neq N/2$.
Since $\Psi(\mathbf{z})=\Psi(s_0\mathbf{z})$, we have 
$\hat{T}_{0}\Psi(\mathbf{z})=-q_{0}\Psi(\mathbf{z})$. 
From the correspondence (\ref{BZ-TY}) and Lemma \ref{lemma-one-act-1}, 
we obtain the desired expression.
One can prove Proposition for $N$ odd in a similar way.
\end{proof}

\section{Laurent polynomial solutions of qKZ equation}
\label{sec-Laurent}
\subsection{Two-boundary case}
\begin{theorem}
\label{theorem-qKZ-2b}
The non-symmetric Koornwinder polynomial $E_{\nu^{J,\pm}}$, 
$J\in\mathbb{Z}_{\ge1}$,  with 
the specialization (\ref{spec}) with $(k,r')=(1,r+1)$ yields a solution of 
the boundary qKZ equation.  
Especially, $\Psi_0=E_{\nu^{J,\pm}}$. 
The parameters have a constraint 
\begin{eqnarray}
\label{const}
\omega_{m}^{\pm mJ/r}q^{\mp(N-1+4J/r)}(q_0q_N)^{\pm1}=\kappa_0\kappa_N.
\end{eqnarray}
\end{theorem}
\begin{proof}
From Proposition \ref{prop-Y-0}, $\Psi_{0}$ is the 
simultaneous eigenfunction of the operator $\hat{Y}_{i}$.
Thus, $\Psi_{0}$ is a non-symmetric Koornwinder polynomial 
with the constraints (\ref{qKZ-non-3}) in Lemma~\ref{lemma-red}. 
From Lemma~\ref{lemma-constraint}, these conditions are 
satisfied by taking $\Psi_{0}=E_{\nu^{J,\pm}}$. 
At $s^{2r}q^{4}=1$, the $\hat{Y}_i$-eigenvalues $y_i$ of $E_{\nu^{J,\pm}}$
is written in terms of $q,q_0$ and $q_N$ explicitly, {\it i.e.,}
$y_{i}=q^{-2(N-i)}s^{2J}q_{0}q_{N}$ for $\nu^{J,+}$ and 
$y_{i}=q^{2(i-1)}s^{-2J}q_{0}^{-1}q_{N}^{-1}$ for $\nu^{J,-}$. 
Together with Lemma~\ref{prop-Y-0}, we obtain the constraint (\ref{const}). 

Given $v$, $v<b_0$, fix a path $p_v:=(\mathbf{v,i})$ 
with $\mathbf{i}=(i_1,\ldots,i_l)$. 
Form Proposition \ref{prop-redqKZ}, $\Psi_{v}$ is written as 
\begin{eqnarray*}
\Psi_{v}=\hat{e_l}\ldots\hat{e}_1\Psi_{0}
+\sum_{w>v}c_{v,w}\Psi_{w}
\end{eqnarray*}
with $c_{v,w}\in\mathbb{K}$.
From Proposition~\ref{prop-diagram} and the definition 
of $\Gamma(\nu^{J,+})$, the function $\Psi_{v}$ is 
characterized by the non-symmetric Koornwinder polynomial 
$E_{\nu}$ with $\nu:=\varphi_{\pm}(v)$, that is 
\begin{eqnarray}
\label{eqn-psi-koorn}
\Psi_{v}\propto E_{\nu}+
\sum_{\lambda}c_{\lambda\mu}E_{\lambda},
\end{eqnarray}
where $\lambda\preceq\mu$ (resp. $\lambda\succeq\mu$) 
for $\varphi_{+}$ (resp. $\varphi_{-}$). 
It remains to show that $e_0\Psi=\hat{e}_0\Psi$. 
From Theorem~\ref{thrm-AHA}, Proposition~\ref{prop-dim} and 
Eqn.(\ref{eqn-psi-koorn}), $\Psi_{v}$, $v\in\{\pm\}^{N}$,  
form the bases of the space $I(\nu^{J,+})$. 
Therefore, $\hat{e_0}\Psi_{v}$ is uniquely written in terms of 
$\Psi_{w}$, that is, 
\begin{eqnarray*}
\hat{e}_0\Psi_{v}=\sum_{w}g_{vw}\Psi_{w},
\end{eqnarray*} 
where $g_{vw}\in\mathbb{K}$ at $s^{2r}q^{4}=1$. 
This $g_{vw}$ is nothing but the matrix representation of 
$e_0$.
This implies that the functions $\Psi_{v}$ satisfy the boundary 
qKZ equations. This completes the proof.
\end{proof}

We define the action of $\tau_i$, $1\le i\le N$, on 
$\mathbf{z}=(z_1,\ldots,z_{N})\in(\mathbb{C}^*)^{N}$  by 
$\tau_{i}:z_{i}\mapsto z_{i}^{-1}$. 
From Eqns.(\ref{qKZ-factor1}) and (\ref{qKZ-factor2}), 
we have 
\begin{eqnarray}
\label{qKZ-tau}
\Psi(\tau_i\mathbf{z})
=\check{R}_{i}(1/(z_{i}z_{i+1}))
\ldots\check{R}_{N-1}(1/(z_iz_{N}))K_N(z_i)
\check{R}_{N-1}(z_N/z_{i})\ldots\check{R}_{i}(z_{i+1}/z_{i})
\Psi(\mathbf{z})
\end{eqnarray}
for $1\le i\le N$. 
Let $b_{+}:=(+\ldots+)$ and $\Psi_{+}:=\Psi_{b_+}$. 
There is no $C_v$ such that the expansion $e_iC_v$ contains 
the term $C_{b_+}$.
Thus Eqn.(\ref{qKZ-tau}) for $\Psi_{+}$ is equal to 
\begin{eqnarray*}
\frac{(z_i+q_{N}\zeta_N)(z_i-q_N\zeta_N^{-1})}
{(1+q_N\zeta_Nz_i)(1-q_N\zeta_N^{-1}z_i)}
\prod_{j=i+1}^{N}
\frac{(qz_{i}^{-1}z_{j}^{-1}-q^{-1})(qz_{i}^{-1}z_{j}-q^{-1})}
{(q-q^{-1}z_{i}^{-1}z_{j}^{-1})(q-q^{-1}z_jz_{i}^{-1})}
\Psi_{+}(\mathbf{z})=\Psi_{+}(\tau_i\mathbf{z}).
\end{eqnarray*}
Since we look for the Laurent polynomial solution, 
the function $\Psi_{+}(\mathbf{z})$ vanishes at 
$z_{j}=q^{2}z_{i}$, $z_{j}=q^{-2}z_{i}^{-1}$,  
$z_{i}=-q_N^{-1}\zeta_{N}^{-1}$ and $z_i=q_N^{-1}\zeta_N$ 
for $1\le i<j\le N$.  
Since $(\hat{T}_{i}-q^{-1})\Psi_{+}(\mathbf{z})=0$ for 
$1\le i\le N$,  
$\Psi_{+}(\mathbf{z})$ is written as 
\begin{eqnarray*}
\Psi_{+}=C\prod_{i=1}^{N}z_i^{-(N-2)}
\prod_{1\le i<j\le N}(qz_{i}-q^{-1}z_{j})(qz_{i}z_{j}-q^{-1})
\prod_{i=1}^{N}(1+q_{N}^{-1}\zeta_{N}^{-1}z_{i}^{-1})
(1-q_{N}^{-1}\zeta_{N}z_{i}^{-1})
\widetilde{\Psi}_{+}(\mathbf{z})
\end{eqnarray*}
where $C$ is a constant term and 
$\widetilde{\Psi}_{+}(\mathbf{z})$ satisfies 
$\widetilde{\Psi}_{+}(\mathbf{z})=\widetilde{\Psi}_{+}(s_i\mathbf{z})
=\widetilde{\Psi}_{+}(\tau_i\mathbf{z})$ for $1\le i\le N$.
A candidate for a solution of the minimal degree is $\widetilde{\Psi}_{+}=1$.
The dominant term in $\Psi_{+}$ is $\prod_{i=1}^{N}z_i^{N-i+1}$. 
Therefore, the solution with the minimal degree corresponds to $(r,J)=(1,1)$. 

\subsection{One-boundary case}
\begin{theorem}
The non-symmetric Koornwinder polynomial $E_{\xi^{+}}$ with the specialization 
(\ref{spec}) with $(k,r')=(2,2r+1)$ yields a solution of the one-boundary 
qKZ equation. Especially, $\Psi_{0}=E_{\xi^{0}}$. 
The parameters have a constraint $q_N^{2}=-q$.
\end{theorem}
\begin{proof}
From Proposition~\ref{prop-Y-0-2}, $\Psi_{0}$ is the simultaneous 
eigenfunction of the operator $Y_{i}$, $1\le i\le N$.
Hence, $\Psi_{0}$ is a non-symmetric Koornwinder polynomial satisfying 
Lemma~\ref{lemma-T-1b}. 
The non-symmetric Koornwinder polynomial $E_{\xi^{0}}$ satisfies 
the above criteria under the specialization $s^{4r}q^{6}=1$ and 
$q_{N}^{2}=-q$.
Since $E_{\xi^{0}}\in I_{+}^{(2,2r+1)}$ and $q_N^{2}=-q$, 
one has a polynomial representation of the affine Hecke algebra with
the dimension $\genfrac{(}{)}{0pt}{}{N}{\lfloor (N+1)/2\rfloor}$ in 
$I_{+}^{(2,2r+1)}$ (see Proposition~\ref{prop-dim-I}).
The component $\Psi_{-}$ is written as 
\begin{eqnarray*}
\Psi_{-}=\sum_{\xi}c_{\xi}E_{\xi}
\end{eqnarray*}
where $c_{\xi}\in\mathbb{K}$ is an indeterminate and the sum is taken 
over all admissible $\xi$ with $E_{\xi}\in I_{+}^{(2,2r+1)}$. 
From Lemma~\ref{lemma-red-qKZ-1b}, any component $\Psi_{b}$, 
$b\in\mathcal{B}_{N}$ is written in terms of $\Psi_{-}$.
The component $\Psi_{-}$ satisfies the following vanishing conditions: 
\begin{eqnarray*}
\Psi_{-}|_{z_{j}=q^{2}z_{i}}=0, \quad 1\le i<j\le N.
\end{eqnarray*}
From $\Psi_{0}\propto E_{\xi^{0}}$ and the vanishing conditions, 
the indeterminates $c_{\xi}$ are determined except normalization. 
By a similar argument to Theorem~\ref{theorem-qKZ-2b}, one can 
show that $\hat{T}_{0}\Psi(\mathbf{z})=-q_{0}\Psi(\mathbf{z})$, 
which implies that $\Psi(\mathbf{z})=\Psi(s_0\mathbf{z})$ under 
the specializations $s^{4r}q^{6}=1$ and $q_N^{2}=-q$. 
Therefore, we have a Laurent polynomial solution of the boundary qKZ equation, 
which is characterized by $\xi^{+}$.
This completes the proof.
\end{proof}

\begin{cor}
We have $\Psi_{(N+1)/2}\propto E_{\xi^{1}}$. 
\end{cor}
\begin{proof}
From Proposition~\ref{prop-Y-0-2}, $\Psi_{(N+1)/2}$ is a simultaneous 
eigenfunction of the operators $\hat{Y}_{i}$, $1\le i\le N$. 
The non-symmetric Koornwinder polynomial $E_{\xi^{1}}$ has 
the same eigenvalues as $\Psi_{(N+1)/2}$ under the specializations 
$s^{4r}q^{6}=1$ and $q_{N}^{2}=-q$. 
Further, the actions of $\hat{T}_{i}$ on $\Psi_{(N+1)/2}$ and 
$E_{\xi^{1}}$ are equivalent to each other. 
Thus we have $\Psi_{(N+1)/2}\propto E_{\xi^{1}}$. 
\end{proof}

Since $(\hat{T}_{i}-q^{-1})\Psi_{-}=0$ for $1\le i\le N-1$, we have 
\begin{eqnarray*}
\Psi_{-}=C\prod_{i=1}^{N}z_{i}^{-(N-1)}
\prod_{1\le i<j\le N}(qz_{i}-q^{-1}z_{j})(qz_{i}z_{j}-q^{-1})
\tilde{\Psi}_{-}(\mathbf{z})
\end{eqnarray*}
where $C$ is a constant term and $\tilde{\Psi}_{-}$ satisfies 
$\tilde{\Psi}_{-}(\mathbf{z})=\tilde{\Psi}_{-}(\tau_{i}\mathbf{z})$ 
for $1\le i\le N$.
The candidate for the solution with the minimal degree is 
$\tilde{\Psi}_{-}=1$. 
In this case, the dominant term in $\Psi_{-}$ is given by 
$\prod_{i=1}^{N}z_{i}^{N-i}$. 
The solution corresponding to $(k,r')=(2,3)$ is the one 
with the minimal degree and studied in \cite{deGPya10,PZJ07}.

\bibliographystyle{amsplainhyper} 
\bibliography{biblio}

\end{document}